\documentclass{article}
\usepackage{fullpage}

\usepackage{amsthm}

\newtheorem{theorem}{Theorem}[section]
\newtheorem{lemma}[theorem]{Lemma}
\newtheorem{corollary}[theorem]{Corollary}
\newtheorem{definition}[theorem]{Definition}
\newtheorem{claim}[theorem]{Claim}

\title{Approximating Unrelated Machine Weighted Completion Time Using Iterative Rounding and Computer Assisted Proofs \footnote{The paper will appear in the Proceedings of 2025 Annual ACM-SIAM Symposium on Discrete Algorithms (SODA 2025).}}
\author{Shi Li\thanks{School of Computer Science, Nanjing University, Nanjing, Jiangsu, China. The work was supported by State Key Laboratory for Novel Software Technology, and New Cornerstone Science Laboratory.}}

\date{}

\usepackage[colorlinks = true, citecolor = blue]{hyperref}

\usepackage{amsfonts, amssymb, amsmath}        
\allowdisplaybreaks
\usepackage{graphicx}
\graphicspath{{figs/}}

\usepackage{algorithm}
\usepackage{algpseudocode}
\algtext*{EndIf}
\algtext*{EndFor}
\algtext*{EndWhile}

\usepackage{enumitem}
\setlist{topsep=3pt, itemsep=1pt}

\usepackage{multirow}

\newcommand{\calI}{\mathcal{I}}

\newcommand{\cost}{\mathrm{cost}}
\newcommand{\vol}{\mathrm{vol}}
\newcommand{\old}{{\mathrm{old}}}
\newcommand{\new}{{\mathrm{new}}}
\newcommand{\next}{\mathrm{next}}

\newcommand{\wC}{{\mathrm{wC}}}

\newcommand{\umk}{\mathrm{umk}}
\newcommand{\mk}{\mathrm{mk}}

\newcommand{\R}{{\mathbb{R}}}
\newcommand{\Z}{{\mathbb{Z}}}
\DeclareMathOperator*{\E}{\mathbb{E}}

\begin{document}

\maketitle

	\begin{abstract}  \small\baselineskip=9pt 
		We revisit the unrelated machine scheduling problem with the weighted completion time objective.  It is known that independent rounding achieves a 1.5 approximation for the problem, and many prior algorithms improve upon this ratio by leveraging strong negative correlation schemes. On each machine $i$, these schemes introduce strong negative correlation between events that some pairs of jobs are assigned to $i$,  while maintaining non-positive correlation for all pairs. 
		
		Our algorithm deviates from this methodology by relaxing the pairwise non-positive correlation requirement. On each machine $i$, we identify many groups of jobs. For a job $j$ and a group $B$ not containing $j$, we only enforce non-positive correlation between $j$ and the group as a whole, allowing $j$ to be positively-correlated with individual jobs in $B$. This relaxation suffices to maintain the 1.5-approximation, while enabling us to obtain a much stronger negative correlation within groups using an iterative rounding procedure: at most one job from each group is scheduled on $i$. 
		
		We prove that the algorithm achieves a $(1.36 + \epsilon)$-approximation, improving upon the previous best approximation ratio of $1.4$ due to Harris.  While the improvement may not be substantial, the significance of our contribution lies in the relaxed non-positive correlation condition and the iterative rounding framework.  Due to the simplicity of our algorithm, we are able to derive a closed form for the weighted completion time our algorithm achieves with a clean analysis.  Unfortunately, we could not provide a good analytical analysis for the quantity; instead, we rely on a computer assisted proof. Nevertheless,  the checking algorithm for the analysis is easy to implement, essentially involving evaluation of maximum values of single-variable quadratic functions over given intervals.  Therefore, unlike previous results which use intricate analysis to optimize the final approximation ratio, we delegate this task to computer programs.

	%TODO: think about how to say the relaxed requirement. 
		
	\end{abstract}
	
	\section{Introduction}

The unrelated machine weighted completion time problem is a classic problem in scheduling theory. We are given a set $J$ of $n$ jobs and a set $M$ of $m$ machines. Each job $j \in J$ has a weight $w_j \in \R_{>0}$, and there is a $p_{ij} \in \R_{>0} \cup \{\infty\}$ for every $i \in M$ and $j \in J$, indicating the processing time needed to process $j$ on machine $i$; sometimes $p_{ij}$ is also called the size of $j$ on $i$. The goal of the problem is to schedule the $n$ jobs $J$ on the $m$ machines $M$, so as to minimize the total weighted completion time over all jobs. It is well-known that once we have an assignment $\phi : J \to M$ of jobs to machines, it is optimum to schedule the jobs $j \in \phi^{-1}(i)$ on each $i \in M$ in descending order of $\frac{w_j}{p_{ij}}$, i.e., Smith ratios. Therefore, it suffices to use the function $\phi : J \to M$ to denote a solution, and its weighted completion time is
\begin{align*}
	\sum_{i \in M}\Big(\sum_{j \in \phi^{-1}(i)} p_{ij} w_j + \sum_{\{j, j'\}\subseteq \phi^{-1}(i): j \neq j'} \min\{p_{ij} w_{j'}, p_{ij'} w_j\}\Big).
\end{align*}

The problem was proved to be strongly NP-hard and APX-hard \cite{hoogeveen2001non}, and several mathematical programming based algorithms gave a $1.5$-approximation ratio for the problem \cite{schulz2002scheduling, Skutella01, sethuraman1999optimal}. The algorithms use the following independent rounding framework: solve a linear/convex/semi-definite program to obtain a fractional assignment $(x_{ij})_{i\in M, j\in J}$ of jobs to machines, and then randomly and independently assign each job $j$ to a machine $i$, with probabilities $x_{ij}$.  
A simple example shows that independent rounding can only lead to a 1.5-approximation, even if the fractional assignment is in the convex hull of integral assignments.  It was a longstanding open problem to break the approximation factor of 1.5 \cite{chekuri2004approximation,schulz2002scheduling,kumar2008minimum,sviridenko2013approximating,schuurman1999polynomial}. 

In a breakthrough result,  Bansal et al.\ \cite{BansalSS16} solved the open problem by giving a $(1.5 - 10^{-7})$-approximation for the problem using their novel strong negative correlation rounding scheme.  They solve a SDP relaxation for the problem to obtain the fractional assignment $(x_{ij})_{i \in M, j \in J}$. If two jobs $j$ and $j'$ are assigned to the same machine $i$, then one will delay the other; this happens with probability $x_{ij}x_{ij'}$ in the independent rounding algorithm. To improve the approximation ratio of $1.5$, they define some groups of jobs for each machine $i$. For two distinct jobs $j$ and $j'$, they guarantee the probability that $j$ and $j'$ are both assigned to $i$ is at most $x_{ij}x_{ij'}$, and at most $(1 - \xi)x_{ij}x_{ij'}$ for some absolute constant $\xi > 0$, if $j$ and $j'$ are in a same group. In other words, they introduced strong negative correlation between pairs within the same group while maintaining non-positive correlation between all pairs. 

Since the seminal work of Bansal et  al.\ \cite{BansalSS16}, there has been a series of efforts aimed at improving the approximation ratio using strong negative correlation schemes.  Li \cite{Li20} combined the scheme of Bansal et al.\ with a time-indexed LP relaxation, to give an improved the approximation ratio to $1.5 - 1/6000$.  Subsequently, Im and Shadloo \cite{IS20} developed a scheme inspired by the work of Feige and Vondrak \cite{feige2008allocations}, further improving the approximation ratio to $1.488$.  More recently, Im and Li \cite{IL23} introduced a strong negative correlation scheme by borrowing ideas from the online correlated selection technique for online edge-weighted bipartite matching \cite{FHT20}, and further improved the ratio to $1.45$.  Both of the results use the time-indexed LP relaxation.  The current state-of-the-art result is a 1.4-approximation algorithm given by Harris \cite{Har24}, based on the SDP relaxation and a novel dependent rounding scheme developed by the author.

\subsection{Our Results} While prior better-than-1.5 approximation algorithms leveraged various strong negative correlation schemes, they all share the common principle we highlighted: introducing strong negative correlation between intra-group pairs of jobs, while maintaining non-positive correlation between all pairs.   Our algorithm deviates from this approach by introducing a more relaxed version of the non-positive correlation requirement. Specifically, in some cases, we only require non-positive correlation to hold between a job $j$ and a group as a whole entity, instead of individual jobs in the group: for a machine $i$, conditioned on $j$ being assigned to $i$, the expected total size of jobs in the group assigned to $i$ is at most the unconditional counterpart. With this relaxed notion of non-positive correlation, we can design a novel and simple iterative rounding procedure, that ensures at most one job from each group is assigned to $i$. This gives the strongest possible negative correlation within each group, leading to our improved approximation algorithm: 
	\begin{theorem}
		\label{thm:main} There is a polynomial time randomized $(1.36 + \epsilon)$-approximation algorithm for the unrelated machine weighted completion time scheduling problem for any constant $\epsilon > 0$.
	\end{theorem}
	
	While our improvement over the prior best ratio of $1.4$ may not be substantial, we believe the significance of our contribution lies in the relaxed non-positive correlation condition and the iterative rounding framework.  The algorithm performs one of two operations in each iteration, and terminates with an integral solution when no operations can be performed. In contrast,  previous algorithms need to design various strong negative correlation schemes, often involving some complex procedures. In these schemes, the parameter $\xi$ in the probability $(1-\xi) x_{ij} x_{ij'}$ that two jobs in a same group tends to be small, while our procedure guarantees $\xi = 1$. 
	
	Due to the simplicity of our algorithm, we are able to derive a closed form for the weighted completion time the algorithm provides with a clean analysis.  Unfortunately, we could not provide a good analytical analysis for the ratio between the derived quantity and the LP cost; instead, we rely on a computer assisted proof. Nonetheless,  the checking algorithm for the analysis is easy to implement, essentially involving the evaluation of maximum values of single-variable quadratic functions over given intervals.  This leads to another advantage of our result: while previous results use involved analysis to optimize the approximation ratio, we delegate the tasks of optimizing the verifying the ratio to computer programs.
	
	Finally, it is worth noting two features of our algorithm.  Firstly, it is the first algorithm to leverage the condition that the weights $w_j$'s are independent of machines. Prior algorithms do not use this condition and thus they are capable of solving the more general scheduling problem where we also have machine-dependent weights $w_{ij}$.  Secondly, our algorithm is based on the configuration LP relaxation for the problem, marking the first use of this relaxation in designing improved approximation algorithms for the scheduling problem. 
	
	\subsection{Overview of Our Algorithm and Analysis}  For the sake of convenience, we swap job sizes and weights to obtain a scheduling instance with machine-independent job sizes $p_j$, but machine dependent weights $w_{ij}$.  It is a simple observation that the swapping does not change the instance.   We solve the configuration LP to obtain a fractional assignment $(z_{ij})_{i \in M, j \in J}$.
	
	With the relaxed non-positive correlation condition, we design our iterative rounding algorithm as follows.  Let $G = (M, J, E)$ be the support bipartite graph of $z$. The volume of an edge $ij$ is defined as $p_j z_{ij}$, and we use $\vol(E')$ to denote the total volume edges in $E'$ for every $E' \subseteq E$.   We partition jobs into classes based on their sizes geometrically, with a random shift: a job $j$ is in $J_k, k \in \Z$, if we have $p_j \in [\beta \rho ^k, \beta\rho^{k+1})$, for some $\rho > 1$ and a randomly chosen $\beta \in [1, \rho)$. 
	
	Our iterative rounding algorithm handles each job class $J_k$ separately and independently; henceforth we fix $k$.    For every machine $i$, let $\delta^k_i$ be the set of edges between $i$ and $J_k$. We sort the edges $ij \in \delta^k_i$ in descending order of their $\frac{w_{ij}}{p_j}$ values, i.e., Smith ratios.  We \emph{mark} the first $\rho \beta^k$ volume the edges in this order, or all edges in $\delta^k_i$ if $\vol(\delta^k_i) \leq \beta \rho^k$; the other edges in $\delta^k_i$ are \emph{unmarked}.  Let $\delta^{k, \mk}_i$ be the set of marked edges in $\delta^k_i$. In the actual algorithm, we may need to break some edge into two parallel edges, one marked and other unmarked. For simplicity, we assume this does not happen in the overview.
	
	We maintain a $\bar x_{ij}$ value for every $ij \in E$ between $M$ and $J_k$, initially set to $z_{ij}$. Let $\bar G = (M, J_k, \bar E)$ be the support of $\bar x$.  A crucial property we maintain is that $\sum_{e \in \delta^{k, \mk}_i} \bar x_e p_e$ does not change (we define $p_e = p_j$ if $e = ij$), until when $|\delta^{k, \mk}_i \cap \bar E| = 1$ happens.  Two operations respecting the property can be defined: a rotation operation over a cycle of marked edges in $\bar E$, and a shifting operation along a ``pseudo-marked-path'' in $\bar E$.  Both operations are random and respect the marginal probabilities, and the condition that the every job is assigned to an extent of 1. As the former operation is standard, we only describe the latter. A pseudo-marked-path consists of a simple path $(j_1, i_1, j_2, i_2 \cdots, i_{t-1}, j_t)$ of edges in $\bar E$ between two jobs $j_1$ and $j_t$, with one edge $i_0j_1$ at the beginning and one edge $j_ti_t$ in the end.  The edge $i_0j_1$ may be unmarked, in which case we allow $i_0 \in \{i_1, i_2, \cdots, i_{t-1}\}$, or the unique edge in $\delta^{k, \mk}_{i_0} \cap \bar E$. The same requirement  is imposed on $j_ti_t$.  (See Definition~\ref{def:pseudo-marked-parth} for the formal definition.) Naturally, we shift $\bar x$-values on the path in one of the two directions randomly.  Each operation will remove at least 1 edge from $\bar G$, and the procedure terminates when we have no cycles of marked edges or pseudo-marked-paths in $\bar E$.  We prove that when this happens, the $\bar x$-vector is integral, which gives the assignment of $J_k$ to $M$. 
	%\footnote{A minor remark: In the formal description of the algorithm, we use $x$ for the initial $x$-vector and it does not change; as there might be parallel edges, $x$ may be slightly different from $z$. We use $\bar x$ to denote the vector that undergoes updates during the iterative rounding procedure. The volume of an edge is defined w.r.t the $x$-vector, and so it also never changes. }
	 
		Now we focus on the correlation between two edges in $\delta^k_i$. Two unmarked edges have a non-positive correlation. Positive correlations may be introduced between an unmarked edge and a marked one by the shifting operation along a pseudo-marked path. It may happen that $i_1j_1$ is unmarked and  $i_1$ is the same as some $i_o$ on the path. Then $i_1j_1$ and $i_oj_o$ will be positively related.  However, the quantity $\sum_{e \in \delta^{k, \mk}_i} \bar x_e p_e$ does not change;  that is, $j$ is non-positively correlated with $\delta^{k, \mk}_i$ as a whole.  This is where the relaxed non-positive correlation condition comes into play.  As all the marked edges precede unmarked ones in the final schedule, this suffices to preserve the $1.5$-approximation ratio.   The improvement over 1.5 arises from the strong negative correlation among marked edges. We maintain $\sum_{e \in \delta^{k, \mk}_i} \bar x_e p_e \leq \beta \rho^k$ until $\delta^{k, \mk}_i \cap \bar E$ becomes a singleton. As every class-$k$ job has size at least $\beta\rho^k$, at most one edge in $\delta^{k, \mk}_i$ will be chosen in our final schedule, resulting in the strongest possible negative correlation within $\delta^{k, \mk}_i$.

	For the analysis, we derive a relatively simple closed form for the expected total weighted completion time.  We can separately focus on each machine $i$, a threshold $\sigma$ on the Smith ratio, and all the jobs $j$ with $\frac{w_{ij}}{p_i} \geq \sigma$, denoting them as $J'$. The term leading to the improvement over $1.5$ is precisely expressed as $\frac12 \sum_{k} \big(\vol(\delta^{k, \mk}_i \cap J')\big)^2 = \frac12 \sum_{k} \big(\min\big\{\vol\big(\delta^k_i \cap J'), \beta \rho^k\big\}\big)^2$. While this term is simple, we still encounter difficulties in providing a good analytical analysis of our ratio. 
	
	Instead, we provide a computer assisted proof with $\rho = 2$. We only focus on three most important job classes $k$, disregarding the savings from others. The ratio is captured by some mathematical program with a maximization objective.  Depending on whether $\vol(\delta^{k, \mk}_i) = \vol(\delta^{k}_i)$ or $\vol(\delta^{k, \mk}_i) = \beta\rho^k_i$ for each of the three classes $k$, we further partition the program into three  sub-programs. By specifying some Lagrangian multipliers for the sub-programs, we establish upper bounds on the sub-programs, and thus the original program. This  results in our $(1.36 + \epsilon)$-approximation ratio. A computer program verifies the validity of all the multipliers and upper bounds. The checker algorithm is easy to implement; essentially, it involves checking hundreds of quadratic function problems, each asking for the maximum value of a single-variable quadratic function over an interval.

	\subsection{Other Related Work} The unrelated machine weighted completion time problem can be solved in polynomial time if jobs have the same weights or if $p_{ij} \in \{1, \infty\}$ for every $ij$ pair \cite{horn1973minimizing,bruno1974scheduling}. When $m = 1$, the Smith rule gives the optimum schedule. Additionally, when $m = O(1)$ there exists a polynomial-time approximation scheme (PTAS)  \cite{lenstra1990approximation}.
	For the problem where weights also depend on the machines and $w_{ij} = p_{ij}$ for every $i \in M, j \in J$,  Kalaitzis et al.\ developed a 1.21-approximation algorithm \cite{Ola2017unrelated}; in this case, all jobs have the same Smith ratios on all machines.  For the cases where the machines are identical, or uniformly related, the problem remains NP-hard \cite{garey2002computers} but admit PTASes  \cite{afrati1999approximation,skutella2000ptas,chekuri2001ptas}.
	\medskip
	
	\noindent{\bf Organization}\ \ The rest of the paper is organized as follows.  In Section~\ref{sec:alg}, we describe our iterative rounding algorithm for the scheduling problem, and show that it terminates with an integral assignment. In Section~\ref{sec:bound-wC}, we establish an upper bound on the weighted completion time given by our algorithm. In Section~\ref{sec:ratio}, we compare the upper bound and the LP cost to obtain our approximation ratio. It suffices to focus on a fixed machine $i$, and the set of jobs whose Smith ratios are above some threshold on $i$.  We conclude with some open problems in Section~\ref{sec:discuss}. 
	
	\section{Iterative Rounding for Unrelated Machine Weighted Completion Time Scheduling} \label{sec:alg}
	
	We describe our iterative rounding algorithm for the unrelated machine weighted time scheduling problem. 
		
	\subsection{Swapping Job Sizes and Weights}
	
		At the very beginning of our algorithm, we swap job sizes and weights so that in the new instance, a job $j$ has a fixed size $p_j$, but machine-dependent weights $w_{ij}$'s.  To achieve this goal, we consider the general problem where both sizes and weights depend on the machines: We are given a weight $w_{ij} > 0$ and a processing time $p_{ij} > 0$ for every $i \in M$ and $j \in C$. We need to find an assignment $\phi: J \to M$ of jobs to machines.  On every machine $i$, we consider the schedule of jobs $\phi^{-1}(i)$ using the Smith rule, with job processing times $(p_{ij})_{j \in \phi^{-1}(i)}$ and weights $(w_{ij})_{j \in \phi^{-1}(i)}$. This gives us a weighted completion time for machine $i$; the objective we try to minimize is the sum of this quantity over all machines $i$. 
		
		A simple but useful observation we use is the following:  
		\begin{lemma}
			Consider two instances $\calI$ and $\calI'$ of the general weighted completion time problem on the same set $M$ of machines and the same set $J$ of jobs.  $\calI$ has processing times $p_{ij}$ and weights $w_{ij}$, and $\calI'$ has processing times $p'_{ij}$ and weights $w'_{ij}$.  For every $i \in M, j \in J$, we have $p_{ij} = w'_{ij}$ and $w_{ij} = p'_{ij}$. Then for every $\phi: J \to M$ , the cost of $\phi$ for $\calI$ is the same as that for $\calI'$.
		\end{lemma}
		\begin{proof}
			For every machine $i$, define $\prec_i$ to be the total order over $J$ w.r.t the Smith ratios on machine $i$ in instance $\calI$. That is, for every $j, j' \in J, j \neq j'$,  $j \prec_i j'$ holds if and only if $j$ should be processed before $j'$ if they are both assigned to $i$ in the instance $\calI$.  This implies $\frac{w_{ij}}{p_{ij}} \geq \frac{w_{ij'}}{p_{ij'}}$; for two jobs with the same Smith ratio, we break the tie arbitrarily.  Define $\prec'_i$ in the same way but for the instance $\calI'$. Then, we can guarantee $j \prec_i j'$ if and only if $j' \prec'_i j$ for every $i \in M, j, j' \in J, j \neq j'$. 
			
			The cost of any $\phi: J \to M$ w.r.t $\calI$ is 
			\begin{align*}
				\sum_{i \in M, j, j' \in \phi^{-1}(i): j = j'\text{ or } j \prec_i j'} p_{ij} w_{ij'} \quad &= \quad \sum_{i \in M, j, j' \in \phi^{-1}(i): j = j'\text{ or } j' \prec'_i j} w'_{ij} p'_{ij'}\\
				&=\sum_{i \in M, j, j' \in \phi^{-1}(i): j = j'\text{ or } j \prec'_i {j'}} p'_{ij} w'_{ij'}\ .
			\end{align*}
			which is precisely the cost of $\phi$ w.r.t $\calI'$. 
		\end{proof}
		
		Therefore, swapping the job sizes and weights does not change the instance.  In the standard unrelated machine weighted completion time problem, the weights are machine-indepdent. By the lemma, it is equivalent to the problem where weights may depend on machines but job sizes are machine-independent.   
		
		From now on, we focus on such an instance.  Every job $j$ has a fixed size $p_j > 0$, but machine-dependent weights $(w_{ij})_{i \in M}$.  We remark that the swapping is only for the sake of convenience, as one could obtain an equivalent algorithm without swapping them. As in \cite{BansalSS16} and \cite{Har24}, we shall partition jobs into classes according to their sizes, and machine-independent sizes will ensure a global partition.  Without the swapping, we have to partition the jobs according to their weights, which is incompatible with prior algorithms.  Also, our algorithm crucially depends on that sizes are machine-independent after swapping. So, unlike prior algorithms, ours does not work for the general scheduling problem. 
		
		We shall use $p(J') := \sum_{j \in J'} p_j$ for every $J' \subseteq J$ to denote the total size of jobs in $J'$. 
		
		%. In many previous algorithms, we partition jobs according to their sizes.    As sizes and weights are symmetric to each other, the swapping is only for the sake of convenience. %TODO: why..  Groups are defined w.r.t their sizes. It is more natural to partition the jobs using their sizes rather than their weights. 
		
		%TODO: define $\prec$? On every machine $i$, we define a total order $\prec_i$ over $J$ using Smith ratios as in the proof of Lemma xx, breaking ties arbitrarily. Define $\preceq_i$ ?
	
	%TODO: A lemma:   

	\subsection{Configuration LP}
	We describe the configuration LP for the problem. As usual, a configuration $f$ is a subset of $J$. For every $i \in M, f \subseteq J$, we define $\cost_i(f)$ to be the total weighted completion time of scheduling jobs $f$ on machine $i$ optimally, i.e., using the Smith rule.  Then in the configuration LP, we have a variable $y_{if}$ for every machine $i \in M$ and configuration $f \subseteq J$, indicating if the set of jobs assigned to $i$ is precisely $f$.  The LP is as follows:
	\begin{equation}
		\min \quad \sum_{i \in M} \sum_{f} y_{if} \cdot \cost_i(f) \label{cfLP}
	\end{equation}\vspace*{-20pt}
	
	\noindent\begin{minipage}[t]{0.45\textwidth}
		\begin{align}
			\sum_{f} y_{if} &=1 &\quad &\forall i \in M \label{LPC:one-configuration}\\
			\sum_{i \in M} \sum_{f \ni j} y_{if} &= 1 &\quad &\forall j \in J \label{LPC:job-scheduled}
		\end{align}	
	\end{minipage}\hfill
	\begin{minipage}[t]{0.45\textwidth}
		\begin{align}
			y_{if} &\geq 0 &\quad &\forall i \in M, j \in J \label{LPC:non-negative}
		\end{align}	
	\end{minipage}\medskip
	
	\eqref{LPC:one-configuration} requires that for every machine $i\in M$, we choose exactly one configuration $f$. \eqref{LPC:job-scheduled} requires that every job is covered by exactly one configuration across all machines. \eqref{LPC:non-negative} is the non-negativity constraint. We pay a cost of $\cost_i(f)$ on $i$ if the set of jobs assigned to $i$ is $f$, thus we have the objective \eqref{cfLP}.
	
	Though this LP has exponential number of variables, Sviridenko and Wiese \cite{sviridenko2013approximating} showed that it can be solved approximately within a factor of $1+\epsilon$, for any constant $\epsilon > 0$. So, we assume we are given a $(1+\epsilon)$-approximate solution $y$ to the LP; in particular, the result of \cite{sviridenko2013approximating} suggests that the number of non-zeros in $y$ is small.  
	
	For every $i \in M$ and $j \in J$, we let $z_{ij} := \sum_{f \ni j} y_{if}$ denote the fraction of job $j$ assigned to machine $i$.  We only use $z$ variables in the rounding algorithm, while reserving the $y_{if}$ variables for the analysis. \medskip
	
	\noindent{\bf Remark}\ \ We remark that the configuration LP could be replaced by the time-indexed LP relaxation in \cite{Li20}, where we have variables $x_{ijs}$ indicating if job $j$ is scheduled on machine $i$, with starting time $s$. A solution for the configuration LP can be derived from the time-index LP with one caveat: a configuration now is a \emph{multi-set} of jobs. This does not pose a significant challenge to our algorithm and analysis, although we need to be careful in notations when handling multi-sets.

	\subsection{Partition of Jobs into Classes and Construction of Bipartite Graph $G$ with Marked and Unmarked Edges}
	Let $\rho$ be a constant which will be set to $2$ later; for future reference we present the algorithm and a part of the analysis for a general $\rho > 1$. We randomly choose a $\beta \in [1, \rho)$ such that $\ln \beta$ is uniformly distributed in $[0, \ln \rho)$.  For an integer $k$, we say a job $j$ is in class $k$ if $p_j \in [\beta \rho^k, \beta \rho^{k+1})$.  Let $J_k$ be the set of jobs in class $k$; notice that it depends on the randomly chosen $\beta$. Once $\beta$ is chosen, we have a global partition $(J_k)_k$ of jobs, independent of the machines.
	
	We construct a bipartite multi-graph $G = (M, J, E)$ between $M$ and $J$, along with  a vector $x \in (0, 1]^E$. Each edge in $E$ is either \emph{marked} or \emph{unmarked}.  The procedure is formally described in Algorithm~\ref{alg:construct-G}.  
	
	\begin{algorithm}
		\caption{Construction of $G$, and vector $x \in (0, 1]^E$}
		\label{alg:construct-G}
		\begin{algorithmic}[1]
			\State let $E \gets \emptyset$
			\For{every $k$ with $J_k \neq \emptyset$ and every $i \in M$}
				\State $v \gets 0$
				\State sort jobs in $J_k$ in descending order of the Smith ratios on machine $i$ (i.e., values of $\frac{w_{ij}}{p_j}$)
				\For{every job $j \in J_k$ in the order, \textbf{if} $z_{ij} > 0$ \textbf{then}}
					\If{$v + p_j z_{ij} \leq \beta \rho^k$}
						\State add to $E$  a \emph{marked edge} $e$ between $i$ and $j$ with $x_e = z_{ij}$
					\ElsIf{$v \geq \beta \rho^k$}
						\State add to $E$  an \emph{unmarked edge} $e$ between $i$ and $j$ with $x_e = z_{ij}$
					\Else\Comment{We have $v < \beta \rho^k < v + p_j z_{ij}$}
						\State add to $E$  a \emph{marked edge} $e$ between $i$ and $j$ with $x_e = \frac{\beta \rho^k - v}{p_j}$ \label{step:G-partial-mark}
						\State add to $E$ an \emph{unmarked edge} $e'$ between $i$ and $j$ with $x_{e'} = z_{ij} - x_e = \frac{v + p_j z_{ij} - \beta \rho^k}{p_j}$ \label{step:G-partial-unmark}
					\EndIf
					\State $v \gets v + p_j z_{ij}$
				\EndFor
			\EndFor
			\State \Return $G := (M, J, E)$ and $x \in (0, 1]^E$
		\end{algorithmic}
	\end{algorithm}
	\medskip
	
	\noindent{\bf Notations}\ \ With $G = (M, J, E)$ and $x \in (0, 1]^E$ constructed, we define the following notations.  For every $e = ij \in E$, we define $p_e = p_{j}$ to be the size of its incident job, and $\sigma_e = \frac{w_{ij}}{p_j}$ be the Smith ratio. For every $i \in M$ (resp., $j \in J$), let $\delta_i$ (resp., $\delta_j$) be the set of edges in $E$ incident to $i$ (resp., $j$).  For every job class $k$, let $E_k$ be the set of edges in $E$ between $M$ and $J_k$, $G_k = (M, J_k, E_k)$ and $x^{(k)} \in (0, 1]^{E_k}$ be the $x$ vector restricted to $E_k$.  For every $k$ and $i \in M$, let $\delta^k_i = \delta_i \cap E_k$ be incident edges of $i$ in $E_k$.   We use superscripts ``$\mk$'' and ``$\umk$'' to denote the restrictions to marked and unmarked edges respectively: $E^\mk$ and $E^\umk$ are the sets of marked and unmarked edges in $E$ respectively. $\delta^\mk_i:=E^\mk \cap \delta_i, \delta^{k, \mk}_i:=E^\mk \cap \delta^k_i$ and $\delta^\mk_j = E^\mk \cap \delta_j$; the same rule applies to unmarked edges as well.  %Notice that all these notations depend on the value of $\beta$.   	
	\begin{definition}
		 For every $e \in E$,  the \emph{volume} of $e$ is defined as $\vol(e) := x_ep_e$. For a subset $E' \subseteq E$ of edges, we define $\vol(E') = \sum_{e \in E'} \vol(e)$.
	\end{definition}
	
	We elaborate on Algorithm~\ref{alg:construct-G} by describing it using an alternative way. For every $i \in M, j \in J$ with $z_{ij} > 0$, we create an edge $ij$ with $x_{ij} = z_{ij}$.  For every job class $k$, and on every machine $i$, we sort the edges $e \in \delta^k_i$ according to $\sigma_e$ in descending order. We find an integer $t$ such that the total volume of the first $t$ edges in this order is exactly $\beta \rho^k$; assume for now $t$ exists. We mark the $t$ edges; the other edges in $\delta^k_i$ are unmarked. If the integer $t$ does not exist, we only mark a portion of some edge $e$ in the order. So, in this case, we break $e$ into two parallel edges,  one marked and the other unmarked, and split the $x$-value (Steps~\ref{step:G-partial-mark} and \ref{step:G-partial-unmark}), so that the volume of marked edges is exactly $\beta \rho^k$. In case $\vol(\delta^k_i) \leq \beta \rho^k$, all edges in $\delta^k_i$ are marked. 	See Figure~\ref{fig:G} for an illustration of the construction of $G$. 
	
	\begin{figure}
		\centering
		\includegraphics[width = 0.5\textwidth]{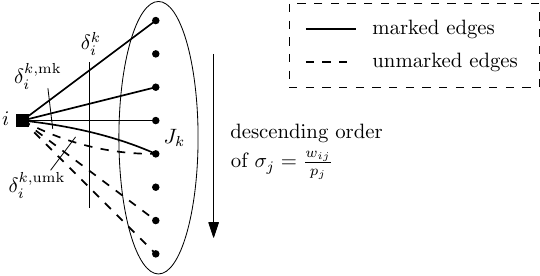}
		\caption{The edges in $G$ between $i$ and $J_k$ for the case $\vol(\delta^k_i) \geq \beta \rho^k$, where we have $\vol(\delta^{k, \mk}_i) = \beta \rho^k$. }
		\label{fig:G}
	\end{figure}
		
	So, the total $x$-value of edges between any $j \in J$ and $i \in M$ is precisely $z_{ij}$.  This implies $\sum_{e \in \delta_j} x_e = 1$ for every $j \in J$. As we are dealing with multi-graphs, when we use $ij$ to denote an edge, we assume we know its identity. The following claim is straightforward:
	\begin{claim}
		\label{claim:vol-marked-k-i}
		For every job class $k$ and machine $i \in M$, we have $\vol(\delta^{k,\mk}_i) = \min\big\{\vol(\delta^k_i), \beta \rho^k \big\}$.
	\end{claim}
	The threshold $\beta \rho^k$ for $\vol(\delta^{k, \mk}_i)$ was carefully chosen as the lower bound for  the size of a class-$k$ job. 
	
	\subsection{Iterative Rounding}
	We describe our randomized iterative rounding procedure to assign jobs to machines. It handles jobs in different classes separately and independently. So, throughout this section, we fix a job class $k$, and show how to assign jobs $J_k$.

	\begin{definition}
		\label{def:pseudo-marked-parth}
		Let $\bar G = (M, J_k, \bar E \subseteq E_k)$ be a spanning subgraph of $G_k$.  A pseudo-marked-path in $\bar G$ is a (not-necessarily-simple) path $(i_0, j_1, i_1, \cdots, j_t, i_t)$ of distinct edges in $\bar G$, where $i_0, i_1, \cdots, i_t \in M$ and $j_1, j_2, \cdots, j_t \in J_k$, satisfying the following properties.  
		\begin{itemize}
			\item The sub-path $(j_1, i_1, …, i_{t-1}, j_t)$ is simple, and all the edges on it are marked. (It is possible that $t = 1$, in which case the sub-path consists of a single job.)
			\item The edge $i_0j_1$ on the path is either unmarked, or the only edge in $\delta^{k, \mk}_{i_0} \cap \bar E$. 
			\item The edge $j_ti_t$ on the path is either unmarked, or the only edge in $\delta^{k, \mk}_{i_t} \cap \bar E$. 
		\end{itemize}
	\end{definition}

	We remark that in the case where $i_0j_1$ is unmarked, it is possible that $i_0 \in \{i_1, i_2, \cdots, i_{t-1}\}$.  So, the path may not be simple.  However, in the other case where $i_0j_1$ is the only edge in $\delta^\mk_{i_0} \cap \bar E$, we have $i_0 \notin \{i_1, i_2, \cdots, i_{t-1}\}$, as every machine in the set has at least 2 incident marked edges in $\bar E$.  The same argument can be made to the last edge $j_ti_t$. See Figure~\ref{fig:pseudo-marked} for an illustration of a pseudo-marked path. 
	\medskip
	
	\begin{figure}
		\centering
		\includegraphics[width=0.8\textwidth]{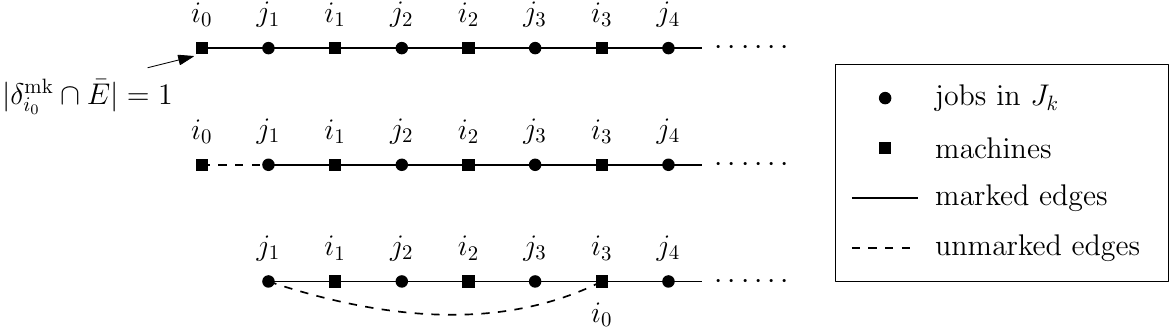}
		\caption{The three cases for the machine $i_0$ on a pseudo-marked path in $\bar G = (M, J^k, \bar E)$. They also apply to the machine $i_t$. }
		\label{fig:pseudo-marked}
	\end{figure}

%TODO: many figures: one for marked and unmarked edges; the other for pseudo-marked-path 
		
	%In the former case, it is possible that $i_0$ is the same as some machine in $\{i_1, i_2, \cdots, i_{t-1}\}$. But this can not happen in the latter case, as $i_0$ has only one incident marked edge. 
	
	\begin{algorithm}[h]
		\caption{Iterative Rounding Procedure for Job Class $k$}
		\label{alg:iter}
		\begin{algorithmic}[1]
			\State $\bar x \gets x^{(k)}, \bar E \gets E_k, \bar G \gets G_k$ \Comment{We maintain that $\bar G = (M, J_k, \bar E)$ and $\bar E$ is the support of $\bar x \in [0, 1]^{E_k}$}
			\While {\textbf{true}}
				\If{we can find a simple cycle $(i_0, j_1, i_1, j_2, i_2, \cdots, j_t, i_t = i_0)$ of marked edges in $\bar G$, or a semi-marked path $(i_0, j_1, i_1, \cdots, j_t, i_t)$ in $\bar G$}
					\State define a non-zero vector $a \in \R^{E_k}$ as in the text \label{step:iter-define-a}
					\State choose the largest $\theta > 0$ such that $\bar x + \theta a$ is non-negative
					\State choose the largest $\theta' > 0$ such that $\bar x - \theta' a$ is non-negative
					\State with probability $\frac{\theta'}{\theta + \theta'}$, let $\bar x \gets \bar x + \theta a$; with the remaining probability $\frac{\theta}{\theta + \theta'}$, let $\bar x \gets \bar x - \theta' a$
					\For{every $e \in \bar E$ with $\bar x_e = 0$}: remove $e$ from $\bar E$ and $\bar G$ \EndFor
				\Else \ \textbf{break}
				\EndIf
			\EndWhile
		\end{algorithmic}
	\end{algorithm}
		
		With the definition, we can describe the iterative rounding procedure for $J_k$. The pseudo-code is given in Algorithm~\ref{alg:iter}.  We need to describe how to define the non-zero vector $a \in \R^{E_k}$ in Step~\ref{step:iter-define-a}.  Suppose the structure (cycle or pseudo-marked-path) we found in $\bar G$ is $(i_0, j_1, i_1, j_2, i_2, \cdots, j_t, i_t = i_0)$ in $\bar G$. We define a non-zero vector $a \in \R^{E_k}$ so that the following conditions hold. 
	\begin{itemize}
		\item Every edge $e \in \bar E$ not on the structure has $z_e = 0$. 
		\item For every $o \in [t]$, we have $a_{i_{o-1}j_o} + a_{j_oi_o} = 0$.
		\item If the structure is a cycle, then for every $o \in [t]$, we have $a_{j_oi_o} p_{j_o} + a_{i_oj_{o+1}} p_{j_{o+1}} = 0$, where for convenience we assume $j_{t+1} = j_1$. 
		\item If the structure is a pseudo-marked-path, then for every $o \in [t - 1]$, we have $a_{j_oi_o} p_{j_o} + a_{i_oj_{o+1}} p_{j_{o+1}} = 0$. 
	\end{itemize}
	In words, the support of $a$ is structure. For a job $j$ on the structure, the sum of $a$-values of its two incident edges on the cycle is $0$. If the structure is a cycle, then for any machine $i$ on the cycle, the sum of $a_{ij}p_j$ over its two incident edges $ij$ on the cycle is $0$.   If the structure is a pseudo-marked-path, we require the equality to hold for machines in $\{i_1, i_2, \cdots, i_{t-1}\}$ and its two incident marked edges on the path. It is easy to see that such $a$ exists and is unique up to scaling.  
	
	The properties for the $a$ vector constructed in either case are summarized in the following claim. 
	\begin{claim}
		\label{claim:properties-of-a}
		The non-zero $a$-vector constructed in Step~\ref{step:iter-define-a} of some iteration satisfies the following properties:
		\begin{enumerate}[label=(\ref{claim:properties-of-a}\alph*), leftmargin=*]
			\item For every $j \in J_k$, we have $\sum_{e \in \delta_j} a_j = 0$.
			\item For every $i \in M$, if $|\delta^{k, \mk}_i \cap \bar E| \geq 2$ at the beginning of the iteration, then $\sum_{e \in \delta^{k, \mk}_i} a_e p_e = 0$.
		\end{enumerate}
	\end{claim}
	\begin{proof}
		The first statement and the second statement for when the structure is a cycle holds trivially. So, we focus on the second statement when the structure we found is a pseudo-marked-path $(i_0, j_1, i_1, \cdots, j_t, i_t)$. If $i \in \{i_1, i_2, \cdots, i_{t-1}\}$, then the equality holds. So assume $i \notin \{i_1, i_2, \cdots, i_{t-1}\}$.  If $i \notin \{i_0, i_t\}$, then all edges in $\delta^{k, \mk}_i$ has $a$ values being $0$. So assume $i \in \{i_0, i_t\}$. Then it must be the case that the edge(s) incident to $i$ on the pseudo-marked-path are unmarked, since otherwise $|\delta^{k, \mk}_i \cap \bar E| = 1$ by the definition of a pseudo-marked-path.  This contradicts the premise of the statement.  So, in this case all edges in $\delta^{k, \mk}_i$ has $a$ values being $0$.
	\end{proof}
	
	We observe that the algorithm will terminate in polynomial number of iterations, as in every iteration of the while loop, we will remove at least one edge from $\bar E$. 

	\subsection{Helper Lemmas}
	Before we conclude that the final $\bar x$ given by Algorithm~\ref{alg:iter} is integral, we prove some useful lemmas. Some of them will be used in the analysis of the approximation ratio in Section~\ref{sec:bound-wC}. Till the end of this section, we fix the random choice $\beta \in [1, \rho)$, which determines $J_k$'s, $G$ and $x$.  We use $\widehat\E[\cdot]$ to denote $\E[\cdot | \beta]$, the expectation condition on a fixed $\beta$.  Then, we fix a job class $k$, and focus on an iteration of the while loop in Algorithm~\ref{alg:iter} for this $k$. We assume the if condition holds.  Let $\bar x^{\old}$ and $\bar x^\new$  be the $\bar x$ values at the beginning and end the iteration respectively.  We use $\E'[\cdot]$ to denote the expectation over the randomness in this iteration. 
	\begin{claim}
		\label{claim:bar-x-marginal}
		For every $e \in E^k$,  we have $\E'[\bar x^\new_e] = \bar x^{\old}_e$. 
	\end{claim}
	\begin{proof}
		This follows from the way we update $\bar x$ in the iteration:  $\bar x^\new_e - \bar x^\old_e$ is $\theta a_e$ with probability $\frac{\theta'}{\theta + \theta'}$, and $-\theta' a_e$ with probability $\frac{\theta}{\theta + \theta'}$. So $\E'[\bar x^\new_e - \bar x^\old_e] = 0$.
	\end{proof}
	
	\begin{claim}
		\label{claim:bar-x-sum-1}
		For every $j \in J_k$, we always have $\sum_{e \in \delta_j}\bar x^\new_e = \sum_{e \in \delta_j}\bar x^\old_e$.
	\end{claim}
	\begin{proof}
		This follows from the first statement of Claim~\ref{claim:properties-of-a}. 
	\end{proof}
	
	\begin{claim}
		\label{claim:non-positive-corr}
		For a machine $i$, and two distinct edges $e, e' \in \delta^{k, \umk}_i$, we have  $\E'[\bar x^\new_e \bar x^\new_{e'}] \leq \bar x^\old_e \bar x^\old_{e'}$. 
	\end{claim}
	\begin{proof}
		If the structure we found in the iteration is a cycle, then $a_e = a_{e'} = 0$ and the equality holds trivially.  So, we assume the structure is a pseudo-marked-path $(i_0, j_1, i_1, \cdots, j_t, i_t)$. If at most one of $e$ and $e'$ is on the path, then the lemma follows from Claim~\ref{claim:bar-x-marginal}. Consider the case $e$ and $e'$ are both on the path. This can only happen if $i_0 = i_t = i$ and $e$ and $e'$ are $i_0j_1$ and $j_ti_t$.  This holds as $\E'[\bar x^\new_e] = \bar x^\old_e$, $\E'[\bar x^\new_{e'}] = \bar x^\old_{e'}$, and $(\bar x^\new_e - \bar x^\old_e)(\bar x^\new_{e'} - \bar x^\old_{e'}) < 0$ with probability 1.
	\end{proof}

	\begin{lemma}
		\label{lemma:mark-vol-no-change}
		Let $i \in M$. If $|\delta^{k, \mk}_i \cap \bar E| \geq 2$ at the beginning of the iteration, then we have $\sum_{e \in \delta^{k, \mk}_i} \bar x^\new_e p_e = \sum_{e \in \delta^{k, \mk}_i} \bar x^\old_e p_e$.
	\end{lemma} %TODO: use e instead of ij.
	
	\begin{proof}
		This follows from the second statement of Claim~\ref{claim:properties-of-a}, and that $\bar x^\new - \bar x^\old$ is $\theta a$ or $-\theta' a$.
	\end{proof}

	\begin{lemma}
		\label{lemma:cross-volume-decreases}
		For a machine $i \in M$, an edge $e \in \delta^{k, \umk}_i$, we have 
		\begin{align*}
			{\textstyle \E'}\Bigg[\sum_{e' \in \delta^{k, \mk}_i}\bar x^\new_e \bar x^\new_{e'} p_{e'}\Bigg] \leq \sum_{e' \in \delta^{k, \mk}_i}\bar x^\old_e \bar x^\old_{e'} p_{e'}.
		\end{align*}
	\end{lemma}
	\begin{proof}
		If the structure we found in the iteration is  a cycle, then $a_e = 0$ as $e$ is unmarked. The lemma follows from Claim~\ref{claim:bar-x-marginal}.  So, we assume the structure we found is a pseudo-marked-path. 
		
		By Lemma~\ref{lemma:mark-vol-no-change}, if at the beginning of the iteration, $|\delta^{k, \mk}_i \cap \bar E| \geq 2$,  then $\sum_{e' \in \delta^{k, \mk}_i} \bar x^\new_{e'} p_{e'} = \sum_{e' \in \delta^{k, \mk}_i} \bar x^\old_{e'} p_{e'}$ always holds.
		\begin{align*}
			{\textstyle \E'}\Bigg[\sum_{e' \in \delta^{k, \mk}_i}\bar x^\new_e \bar x^\new_{e'} p_{e'}\Bigg] &= {\textstyle \E'}\Bigg[\bar x^\new_e\sum_{e' \in \delta^{k, \mk}_i}\bar x^\new_{e'} p_{e'}\Bigg] = {\textstyle \E'} \Bigg[\bar x^\new_e \sum_{e' \in \delta^{k, \mk}_i} \bar x^\old_{e'} p_{e'}\Bigg] \\
			&=  {\textstyle \E'} \Big[\bar x^\new_e\Big] \sum_{e' \in \delta^{k, \mk}_i} \bar x^\old_{e'} p_{e'} = \bar x^{\old}_e\sum_{e' \in \delta^{k, \mk}_i} \bar x^\old_{e'} p_{e'}.
		\end{align*}
		
		Then we consider the case $ |\delta^{k, \mk}_i \cap \bar E| = 1$; assume the unique edge in $\delta^{k, \mk}_i \cap \bar E$ is $e'$.  If at most one of $e$ and $e'$ is on the pseudo-marked-path, the lemma follows from Claim~\ref{claim:bar-x-marginal}. If both of them are on the pseudo-marked-path, denoted as $(i_0, j_1, i_1, \cdots, j_t, i_t)$, then it must be the case that $i_0 = i_t$, and $e$ and $e'$ are $i_0j_1$ and $j_ti_t$. Then in the iteration, one of $a_e$ and $a_{e'}$ is positive and the other is negative. Then, $\E'[\bar x^\new_e \bar x^\new_{e'}] \leq \bar x^\old_e \bar x^\old_{e'}$, which implies the lemma. 
	\end{proof}
	
	By considering all iterations of Algorithm~\ref{alg:iter} and defining martingales (sub-martingales) appropriately, we can prove the following corollary.  Recall that $\widehat\E[\cdot] = \E[\cdot | \beta]$. 
	\begin{corollary}
		\label{coro:algo-iter}
		Let $\bar x$ and $\bar E$ be the $\bar x$ and $\bar E$ at the end of Algorithm~\ref{alg:iter}. The following statements hold: 
		\begin{enumerate}[label=(\ref{coro:algo-iter}\alph*),leftmargin=*]
			\item For every $e \in E^k$, we have $\widehat\E[\bar x_e] = x_e$. 
			\item \label{property:j-x-sum-1} For every $j \in J_k$, we have $\sum _{e \in \delta_j} \bar x_e = \sum_{e \in \delta_j}x_e = 1$.  
			\item \label{property:umk-non-positive-cor} For a machine $i$ and two distinct edges $e, e'\in \delta^{k, \umk}_i$, we have $\widehat \E[\bar x_e\bar x_{e'}] \leq x_e x_{e'}$.
			\item \label{property:marked-2} If $|\delta^{k, \mk}_i \cap \bar E| \geq 2$ for some $i \in M$, then $\sum_{e \in \delta^{k, \mk}_i} \bar x_e p_e = \sum_{e \in \delta^{k, \mk}_i} x_e p_e$.
			\item \label{property:umarked-marked} For a machine $i \in M$, an edge $e \in \delta^{k, \umk}_i$, we have $\widehat \E\big[\sum_{e' \in \delta^{k, \mk}_i} \bar x_e \bar x_{e'} p_{e'}\big] \leq \sum_{e' \in \delta^{k, \mk}_i} x_e x_{e'} p_{e'}$.
		\end{enumerate}
	\end{corollary}
	
	\begin{proof}
		The 5 statements respectively follow Claims~\ref{claim:bar-x-marginal},  \ref{claim:bar-x-sum-1}, and  \ref{claim:non-positive-corr}, and Lemmas~\ref{lemma:mark-vol-no-change} and \ref{lemma:cross-volume-decreases}. 
	\end{proof}
	
%	The following statements hold. 
%		\begin{itemize}
%			\item For a machine $i \in M$, if $|\delta_{\bar G}(i)| \geq 2$, then $\sum _{ij \in \bar E} \bar x_{ij} p_j = \sum _{ij \in E_k} x^{(k)}_{ij} p_j$.
%		\end{itemize}

	\subsection{Algorithm~\ref{alg:iter} Returns an Integral Assignment}
		In this section we prove the following lemma, for a fixed job class $k$:
		\begin{lemma}
			\label{lemma:bar-x-integral}
			When Algorithm~\ref{alg:iter} terminates, we have $\bar x \in \{0, 1\}^{E_k}$. 
		\end{lemma}
		Combined with Property~\ref{property:j-x-sum-1}, we have $\bar x$ gives an assignment of $J_k$ to $M$. Our final assignment $\phi$ can be constructed by considering all job classes $k$. This finishes the description of the whole algorithm.
		\begin{proof}[Proof of Lemma~\ref{lemma:bar-x-integral}]
			When the algorithm terminates, we do not have a cycle of marked edges, or a pseudo-marked-path in $\bar E$.  %TODO: somewhere about how to find.
			So $(M, J_k, E^\mk \cap \bar E)$ is a forest.  We focus on a non-singleton tree $T$ of marked edges in the forest.  Let $J_T$ be the set of jobs in the tree $T$, and $M_T$ be the set of machines. Notice that $J_T$ is incident to at most 1 unmarked edge in $\bar E$: If there are two unmarked edges $ji$ and $j'i'$ in $\bar E$ with $j, j' \in J_T$, then we can take the path from $j$ to $j'$ in $T$, and concatenate it with $ij$ at the beginning and with $j'i'$ at the end. This would give us a pseudo-marked-path. 
			%TODO: need a more general proof. 
			
			First consider the case where there are no unmarked edges incident to $J_T$ in $\bar G$.  Notice that $T$ contains at least 2 leaves. If it contains 2 leaf machines, then the path in $T$ between the two machines would be a pseudo-marked-path. Otherwise, $T$ contains at least one leaf-job $j \in J_k$. Assume $j$ is incident to $i \in M$ in $T$. If $i$ has degree at least $2$ in $T$, then $\sum_{e \in \delta^{k, \mk}_i} \bar x_e p_e = \sum_{e \in \delta^{k, \mk}_i} x_e p_e  = \vol(\delta^{k, \mk}_i) \leq \beta\rho^k$ by Property~\ref{property:marked-2} and Claim~\ref{claim:vol-marked-k-i}. Notice that $p_j \geq \beta \rho^k$ and $ij$ is the only edge incident to $j$ in $\bar E$, which implies $\bar x_{ij} = 1$.  This contradicts that $i$ has degree at least 2 in $T$. So $i$ has degree 1 in $T$ and $T$ consists of the single edge $ij$ with $x_{ij} = 1$. 
						
			Consider the second case where there is exactly one unmarked edge incident to $J_T$ in $\bar G$.  Assume the edge is $e$ and it is incident to $j \in J_T$.   If there is one leaf machine in $T$, then concatenating the path between the leaf machine and $j$ in $T$ and $e$ would give us a pseudo-marked-path.  So $T$ does not have a leaf machine, and this implies $|J_T| \geq |M_T|+1$.  Then again by Property~\ref{property:marked-2} and Claim~\ref{claim:vol-marked-k-i} we have $\sum_{e \in \delta^{k, \mk}_i} \bar x_e p_e \leq \beta\rho^k$ for every machine $i \in M_T$.  Every job in $J_T$ has size at least $\beta \rho^k$, and $J_T$ is incident to only 1 unmarked edge. This can only happen if $T$ is the singleton $(\{j\}, \emptyset)$.
			
			So, all the marked edges $e \in E_k$ have $\bar x_e \in \{0, 1\}$. As every $j \in J_k$ can be incident to at most 1 unmarked edge in $\bar E$ in the end, all the unmarked edges $e \in E^k$ also have $\bar x_e \in \{0, 1\}$. 
		\end{proof}

		\begin{lemma}
			\label{lemma:only-1-in-mkki}
			For two distinct edges $e, e' \in \delta^{i,\mk}_k$, we have $\bar x_e \bar x_{e'} = 0$ at the end of the algorithm.
		\end{lemma}
		\begin{proof}
			By Lemma~\ref{claim:bar-x-marginal} we have $\bar x_e \in \{0, 1\}$ and $\bar x_{e'} \in \{0, 1\}$ at the end of the algorithm. 
			Assume towards the contradiction that $\bar x_e = 1$ and $\bar x_{e'} = 1$ can happen with positive probability. Then by Property~\ref{property:marked-2} and Claim~\ref{claim:vol-marked-k-i}, we have $\sum_{e'' \in \delta^{i, \mk}_k}\bar x_{e''} p_{e''}  = \sum_{e'' \in \delta^{i, \mk}_k}x_{e''} p_{e''}  \leq \beta \rho^k$. This contradicts that $\bar x_e = 1, \bar x_{e'} = 1, p_e \geq \beta\rho^k$ and $p_{e'} \geq \beta \rho^k$.  
		\end{proof}
	
\section{An Upper Bound on Weighted Completion Time on Machine $i$} \label{sec:bound-wC}
	In this section we fix a machine $i$,  rewrite the LP cost on $i$, and prove an upper bound on the expected weighted completion time on $i$ in our solution.  Then in Section~\ref{sec:ratio}, we analyze the ratio between the two quantities.  As the edges $E$ and the vector $x \in (0, 1]^E$ depend on the random choice $\beta$,  in Section~\ref{sec:ratio} it will be more convenient to express the formulations using jobs $J$ and the $z$-vector. So,  we define $\sigma_j := \frac{w_{ij}}{p_j}$ to be the Smith ratio of $j$ on machine $i$ for every $j \in J$. We index $j$ as $[n]$ so that $\sigma_1 \geq \sigma_2 \geq \cdots \geq \sigma_n$. Define $\sigma_{n+1} = 0$ for convenience.  
	
	We first rewrite the LP cost on $i$:
	\begin{lemma}
		\label{lemma:LP-cost-for-i}
		We have 
		\begin{align}
			\sum_f y_{if} \cdot \cost_i(f) = \sum_{j^*\in [n]} (\sigma_{j^*} - \sigma_{j^*+1}) \cdot \frac12\left(\sum_{j  \in [j^*]} z_{ij} p_j^2 + \sum_{f} y_{if} \cdot p^2(f \cap [j^*])\right). \label{equ:LP-cost-for-i}
		\end{align}
	\end{lemma}
	\begin{proof} The left side of the equality is 
		\begin{align*}
			&\quad \sum_{f} y_{if} \sum_{j', j \in f: j' \leq j } \sigma_j p_j p_{j'} = \sum_{f} y_{if} \sum_{j', j \in f, j^* \in [n]: j' \leq j \leq j^*} (\sigma_{j^*} - \sigma_{j^*+1}) p_j p_{j'} \\
			&= \sum_{j^* \in [n]} (\sigma_{j^*} - \sigma_{j^*+1}) \sum_f y_{if}  \sum_{j' \leq j \leq j^*:j', j \in f} p_jp_{j'}\\
			&= \sum_{j^* \in [n]} (\sigma_{j^*} - \sigma_{j^*+1}) \sum_f y_{if}  \cdot \frac12\left( \sum_{j \in f \cap [j^*]} p_j^2 + \Big(\sum_{j  \in f \cap [j^*]}  p_j\Big)^2\right)\\
			&= \sum_{j^*\in [n]} (\sigma_{j^*} - \sigma_{j^*+1}) \cdot \frac12\left(\sum_{j  \in [j^*]} z_{ij} p_j^2 + \sum_{f} y_{if} \cdot p^2(f \cap [j^*])\right).
		\end{align*}
		The last equality used $\sum_{f \ni j} y_{if} = z_{ij}$ for every $j \in [n]$. 
	\end{proof}
	
	Till the end of the section, we fix the random choice $\beta \in [1, \rho)$. Let $\wC^{(i)}$ denote the total weighted completion time of jobs assigned to $i$ in the solution produced by our algorithm. Recall that $\widehat \E[\cdot]$ is $\E[\cdot|\beta]$.   For  a job $j \in [n]$, define $\vol(j) := z_{ij}p_j$. This is equal to $\sum_{e = ij} \vol(e)$. For a set $J' \subseteq [n]$ of jobs, define $\vol(J') = \sum_{j \in J'} \vol(j)$.   The main goal of this section is to prove the following lemma, that bounds the cost on $i$ in our solution: 
	\begin{lemma} \label{lemma:int-cost-on-i}
		The expected weighted completion time on machine $i$ conditioned on $\beta$ is 
		\begin{align}\label{inequ:int-cost-on-i}
			\widehat\E[\wC^{(i)}] \leq \sum_{j^* = 1}^n (\sigma_{j^*} - \sigma_{j^*+1}) \Bigg(\sum_{j \in [j^*]} z_{ij} p_j^2 + \frac12 \cdot \vol^2([j^*]) - \frac12 \sum_k \Big(\min\Big\{ \vol([j^*] \cap J_k), \beta\rho^k \Big\}\Big)^2 \Bigg).
		\end{align}		
	\end{lemma}

	\begin{proof}
		In the proof, we use edges and $x$-vector instead of jobs and $z$-vector.  Recall that $\sigma_e = \frac{w_{ij}}{p_j}$ if $e$ is an edge between $i$ and $j$. We define a total order $\prec$ over the set $\delta_i$ of incident of edges of $i$ in $G$ using their order of creation in Algorithm~\ref{alg:construct-G}: $e \prec e'$ if $e$ is created before $e'$; this implies $\sigma_e \geq \sigma_{e'}$. In particular, for any two parallel edges $e$ and $e'$ with $e$ marked and $e'$ unmarked, we have $e \prec e'$. We define $e \preceq e'$ if $e \prec e'$ or $e = e'$.  For every $e \in \delta_i$, we define $\next(e)$ to be the edge after $e$ in the order $\prec$. For the last edge $e$ in the order, we define $\sigma_{\next(e)} = 0$ for convenience.  We use $[e]:=\{e' \in \delta_i: e' \preceq e\}$ to denote the set of edges before $e$ in the order $\prec$, including the edge $e$ itself.  
		
		With the notations defined, \eqref{inequ:int-cost-on-i} becomes
			\begin{align} \label{inequ:int-cost-on-i-using-edges}
				\widehat \E\big[\wC^{(i)}\big] \leq \sum_{e^* \in \delta_i}(\sigma_{e^*} - \sigma_{\next(e^*)})\Big(\sum_{e \preceq e^*} x_e p_e^2 + \frac12\cdot \vol^2([e^*]) - \frac12 \sum_{k} \big(\min\big\{\vol([e^*] \cap \delta^k_i), \beta\rho^k\big\}\big)^2 \Big).
			\end{align}  

		%TODO: polish the proof after this point. 
		To see the equivalence of the two inequalities, we make two modifications to the right side \eqref{inequ:int-cost-on-i-using-edges}. First, merging two parallel edges $e$ and $e'$ into a new edge $e''$ with $x_{e''} = x_e + x_{e'}$ does not change the quantity. ($e$ and $e'$ are next to each other in the order $\prec$, and thus the new order can be naturally defined.)  This holds as $\sigma_e = \sigma_{e'}, p_e = p_{e'}$.   Then for any $j^* \in [n]$ with no edges between $i$ and $j$ in $G$, we add an dummy edge $e = ij$ of $p_e = p_J, \sigma_e = \frac{w_{ij}}{p_j}, x_e = 0$,  and insert it into the order $\prec$ naturally. The operation does not change \eqref{inequ:int-cost-on-i-using-edges}.  After the two operations, the right side of \eqref{inequ:int-cost-on-i-using-edges} becomes \eqref{inequ:int-cost-on-i}.  So now we focus on the proof of  \eqref{inequ:int-cost-on-i-using-edges}.   \medskip
		
		For every $k$ and every $e \in \delta^k_i$, we define $X_e$ to be the final value of $\bar x_e$ in Algorithm~\ref{alg:iter} for the job class $k$. By Lemma~\ref{lemma:bar-x-integral}, $X_e \in \{0, 1\}$. The total weighted completion time of jobs on machine $i$ is 
		\begin{align}
			\wC^{(i)}&= \sum_{e \in \delta_i} X_e  \sigma_e p_e \sum_{e' \preceq e} X_{e'} p_{e'} = \sum_{e \in \delta_i} X_e p_e \sum_{e^* \succeq e}  (\sigma_{e^*} - \sigma_{\next(e^*)}) \sum_{e' \preceq e} X_{e'} p_{e'} \nonumber\\
			&= \sum_{e^* \in \delta_i} (\sigma_{e^*} - \sigma_{\next(e^*)}) \sum_{e' \preceq e \preceq e^*} X_eX_{e'}p_e p_{e'} = \sum_{e^* \in \delta_i} (\sigma_{e^*} - \sigma_{\next(e^*)}) \left(\sum_{e \preceq e^*} X_e p_e^2 + \sum_{e' \prec e \preceq e^*} X_e X_{e'} p_e p_{e'}\right). \label{equ:break-by-e^*}
			%= \sum_{e^* \in \delta_i} (\sigma_{e^*} - \sigma_{\next(e^*)}) \cdot \frac12\left(\sum_{e \preceq e^*} X_e p_e^2 + \Big(\sum_{e \preceq e^*} X_e p_e\Big)^2\right).
		\end{align}
		The last equality used that $X_e \in \{0, 1\}$, which implies $X_e^2 = X_e$, for every $e \in \delta_i$. We fix an $e^* \in \delta_i$ and upper bound  $\displaystyle \widehat\E\Big[\sum_{e' \prec e \preceq e^*} X_eX_{e'} p_ep_{e'}\Big]$.
		
		Focus on two edges $e'$ and $e$ with $e' \prec e \preceq e^*$.  First, if $e \in \delta^k_i$ and $e' \in \delta^{k'}_i$ for some $k \neq k'$, then we have $\widehat\E[X_eX_{e'}] = \widehat\E[X_e]\widehat\E[X_{e'}] = x_e x_{e'}$ as Algorithm~\ref{alg:iter} for job classes $k$ and $k'$ are independent. So we assume $e'$ and $e$ are both in $\delta^k_i$ for some $k$. If $e$ is marked, so is $e'$.  By Lemma~\ref{lemma:only-1-in-mkki}, $X_eX_{e'} = 0$ always happens conditioned on the given $\beta$. %TODO: a lemma. 
		If both $e$ and $e'$ are unmarked, we have  $\widehat \E[X_e X_{e'}] \leq x_e x_{e'}$ by Property~\ref{property:umk-non-positive-cor}. 
		
		For an edge $e \in [e^*]\cap \delta^{k, \umk}_i$, we have $\widehat \E\big[\sum_{e' \in \delta^{k, \mk}_i}X_e X_{e'} p_{e'}\big] \leq \sum_{e' \in \delta^{k, \mk}_i}x_e x_{e'} p_{e'}$ by Property~\ref{property:umarked-marked}.   Notice that $e' \prec e$ for every $e' \in \delta^{k, \mk}_i$.  So, for an edge $e \in [e^*] \cap \delta^{k, \umk}_i$, we have
		\begin{align*}
			\widehat \E \Bigg[\sum_{e' \prec e} X_e X_{e'}p_e p_{e'} \Bigg] &= \widehat \E \Bigg[\sum_{e' \in [e] \setminus \delta^k_i} X_e X_{e'}p_e p_{e'} \Bigg] + \widehat \E \Bigg[\sum_{e' \prec e: e' \in \delta^{k, \umk}_i} X_e X_{e'}p_e p_{e'} \Bigg] 
			+ \widehat \E \Bigg[\sum_{e' \in \delta^{k, \mk}_i} X_e X_{e'}p_e p_{e'} \Bigg]\\
			&\leq \sum_{e' \in [e] \setminus \delta^k_i} x_e x_{e'}p_e p_{e'} + \sum_{e' \prec e: e' \in \delta^{k, \umk}_i} x_e x_{e'}p_e p_{e'} + \sum_{e' \in \delta^{k, \mk}_i} x_e x_{e'}p_e p_{e'}\\
			&= \sum_{e' \prec e} x_e x_{e'}p_e p_{e'}.
		\end{align*}
	
		For an edge $e  \in [e^*] \cap \delta^{k, \mk}_i$,  we have 
		\begin{align*}
			\widehat \E\Bigg[\sum_{e' \prec e} X_eX_{e'} p_ep_{e'} \Bigg] &= \widehat \E\Bigg[\sum_{e' \in [e] \setminus \delta^k_i}X_eX_{e'} p_ep_{e'}\Bigg]  = \sum_{e' \in [e] \setminus \delta^k_i}x_ex_{e'} p_ep_{e'} = \sum_{e' \prec e} x_ex_{e'}p_e p_{e'} - \sum_{e' \prec e: e' \in \delta^k_i}x_ex_{e'}p_ep_{e'}.
		\end{align*}
		
		Therefore, for a fixed $e^* \in \delta_i$, we have:
		\begin{align*}
			&\quad \widehat \E\Bigg[\sum_{e' \prec e \preceq e^*} X_e X_{e'} p_e p_{e'}\Bigg] \leq \sum_{e' \prec e \preceq e^*} x_ex_{e'} p_e p_{e'} - \sum_k\sum_{e' \prec e \preceq e^*: e, e' \in \delta^{k, \mk}_i} x_ex_{e'}p_ep_{e'}\\
			&= \frac12\Big(\sum_{e \in [e^*]} x_ep_e\Big)^2- \frac12\sum_{e \in [e^*]} x_e^2p_e^2  - \sum_k\left(\frac12\Big(\sum_{e\in[e^*] \cap \delta^{k, \mk}_i} x_ep_e\Big)^2 - \frac12\sum_{e \in [e^*]\cap\delta^{k, \mk}_i} x_e^2p_e^2\right)\\
			&\leq \frac12\Big(\sum_{e \in [e^*]} x_ep_e\Big)^2 - \frac12\sum_k\Big(\sum_{e\in [e^*]\cap\delta^{k, \mk}_i} x_ep_e\Big)^2 = \frac12 \cdot  \vol^2([e^*]) - \frac12\sum_k\Big(\min\Big\{\vol\big([e^*] \cap \delta^k_i\big), \beta\rho^k\Big\}\Big)^2.
		\end{align*}
		The last equality follows from that we mark the first few edges of volume $\beta \rho^k$ in $\sigma^k_i$ in Algorithm~\ref{alg:construct-G} according to the order $\prec$. 
		
		Combining the inequality, \eqref{equ:break-by-e^*} and that $\widehat \E\big[\sum_{e \preceq e^*} X_e p_e^2\big] = \sum_{e \preceq e^*} x_e p_e^2$, we have 
		\begin{align*}
			\widehat \E[\wC^{(i)}] \leq \sum_{e^* \in \delta_i}(\sigma_{e^*} - \sigma_{\next(e^*)})\Bigg(\sum_{e \preceq e^*} x_e p_e^2 + \frac12\cdot \vol^2([e^*]) - \frac12 \sum_k \big(\min\big\{\vol([e^*] \cap \delta^k_i), \beta \rho^k\big\}\big)^2\Bigg), 
		\end{align*}
		which is precisely \eqref{inequ:int-cost-on-i-using-edges}. 
	\end{proof} \bigskip

%TODO: Fix $j^*$ in Section 4, not here. 

\section{Analyzing Approximation Ratio with Computer Assistance} \label{sec:ratio}
	
	In this section,  we prove the $1.36$-approximation ratio of our algorithm for a fixed machine $i \in M$, by comparing the right side of \eqref{equ:LP-cost-for-i} and the expectation of the right side of \eqref{inequ:int-cost-on-i} over all random choices of $\beta$.  Specifically, we fix $j^*$, and compare the term after $(\sigma_{j^*} - \sigma_{j^*+1})$ in \eqref{equ:LP-cost-for-i}, and the counterpart in \eqref{inequ:int-cost-on-i}.
	
	Therefore, we only consider jobs in $[j^*]$ in this section. We restrict all the configurations on $i$ to be subsets of $[j^*]$.  With a slight abuse of notations, we still use $y_{if}$ values to denote the masses of the configurations, after restricting jobs to $[j^*]$. That is, for any $f \subseteq [j^*]$, the new $y_{if}$ value equals the old $\sum_{f' \cap [j^*] = f} y_{if'}$.   Now we set $\rho = 2$ and our analysis in this section is tailored to this value of $\rho$. The main lemma we prove in this section is:
	\begin{lemma} \label{lemma:1.36-for-i-j*} We have
		\begin{align*}
			&\qquad {\E}_\beta\left[\sum_{j \in [j^*]} z_{ij} p_j^2 + \frac12 \cdot \vol^2([j^*]) - \frac12 \sum_k \Big(\min\Big\{ \vol([j^*] \cap J_k), \beta \cdot 2^k \Big\}\Big)^2\right]\\
			&\leq 1.36\times \frac12\Big(\sum_{j  \in [j^*]} z_{ij} p_j^2 + \sum_{f} y_{if} \cdot p^2(f)\Big).
		\end{align*}
	\end{lemma}
	Combine the lemma with Lemma~\ref{lemma:LP-cost-for-i} and \ref{lemma:int-cost-on-i}, we have that the expected weighted completion time of jobs assigned to $i$ is at most $1.36$ times the cost of $i$ in the LP  solution. This leads to our $1.36$-approximation ratio, finishing the proof of Theorem~\ref{thm:main}.

	\subsection{Notations} For a fixed $\beta$, define $t := \beta \cdot 2^k$ for the biggest integer $k$ such that $\beta \cdot 2^k \leq \frac{\vol([j^*])}{2}$.  Therefore, $\frac{\vol([j^*])}{t} \in [2, 4)$. We let $q_j = \frac{p_j}{t}$ for every $j \in [j^*]$ be a scaled processing time of $j$.  Similarly, define $q(J') = \sum_{j \in J'}q_j$ for every $J' \subseteq [j^*]$.  We define
	\begin{align*}
		v_k(f) := \sum_{j \in f: 2^k \leq q_j < 2^{k+1}} q_j, \quad \text{for any configuration } f \subseteq [j^*] \text{ and integer } k,
	\end{align*}\vspace*{-5pt}
	\begin{align*}
		L &:= \frac{\vol([j^*])}{t} = \sum_{j \in [j^*]} z_{ij}q_j = \sum_{f} y_{if} q(f) \in [2, 4), &\quad &Q:=\sum_{j \in [j^*]} z_{ij} q_j^2 = \sum_{f}y_{if} \sum_{j \in f} q_j^2,\\
		F &:= \sum_{f} y_{if} \cdot q^2(f),  &\quad &S:=\sum_k \Big(\min\Big\{\sum_f y_{if} \cdot v_k(f), 2^k\Big\}\Big)^2.
	\end{align*}
	Notice that all the definitions depend on $\beta$, and thus are random. In particular,  $\ln \frac L2$ is also uniformly distributed in $[0, \ln 2)$.  \medskip
	
	With the notations defined, we proceed to the proof of Lemma~\ref{lemma:1.36-for-i-j*}.  Our goal is to define a constant $\alpha_{L} \in [1, 1.5]$ for every $L \in [2, 4)$ so that
		 	\begin{align}
		 		 Q + \frac12 \cdot L^2 - \frac12  \cdot S  &\leq \alpha_L \cdot \frac12\Big(Q + F \Big), \text{for every }\beta \in [1,  2), \text{and}\label{inequ:individual-alpha-L} \\[3pt]
				{\E}_{\beta} \big[\alpha_L\big] &\leq 1.36.  \label{inequ:expect-alpha-L}
			\end{align}
	Multiplying both sides of \eqref{inequ:individual-alpha-L} by $t^2$, and taking the expectation of both sides over $\beta$, we obtain the inequality in 
	Lemma~\ref{lemma:1.36-for-i-j*}. To see this, notice that $Q \cdot t^2 = \sum_{j \in [j^*]} z_{ij} p_j^2,  L^2 \cdot t^2 = \vol^2([j^*])$ and $F \cdot t^2 = \sum_f y_{if}\cdot p^2(f)$. Letting $k'$ be the integer such that $t = \beta \cdot 2^{k'}$, we have 
	 	\begin{align*}
	 		S \cdot t^2 &= \sum_{k} \Big(\min\Big\{\sum_{f}y_{if} \cdot v_{k}(f)\cdot t, \beta \cdot 2^{k + k'}\Big\}\Big)^2 = \sum_{k''} \Big(\min \Big\{\vol([j^*] \cap J_{k''}), \beta \cdot 2^{k''}\Big\}\Big)^2.
	 	\end{align*} 
	 	The second inequality follows by defining $k'' = k + k'$. Also a job $j \in f$ is counted towards $v_k(f)$ if and only if $j \in J_{k''}$, and when it is counted, we have $q_j \cdot t = p_j$.

	Therefore, it remains to set $\alpha_L$ values so that \eqref{inequ:individual-alpha-L} and \eqref{inequ:expect-alpha-L} hold. If we set $\alpha_L = 1.5$ for every $L$, then \eqref{inequ:individual-alpha-L} holds even without the $\frac12 \cdot S$ term on the left-side,  as $Q \leq F$ and $L^2 \leq F$. So we can always set $\alpha_L \leq 1.5$; this recovers the $1.5$-approximation ratio without using the negative correlation. 
	
	\subsection{Capturing $\alpha_L$ using  a Mathematical Program} Till the end of the section, we fix the random variable $\beta$, and thus $L$. We define a size-configuration to be a multi-set $g \subseteq \R_{\geq 0}$ of positive reals with $|g| \leq n$. As the name suggests, a size-configuration is obtained from a configuration by replacing each job with its size.  For every size configuration $g$, we define $v_0(g), v_1(g)$ and $v_2(g)$ to be the sum of elements of $g$ in $[1, 2), [2, 4)$ and $[4, 8)$ respectively.  This corresponds to the definition of $v_k(f)$ for configurations $f$; but we only need to define $v_0(g), v_1(g)$ and $v_2(g)$ for size-configurations $g$. 
	
	Let $\alpha \in [1, 1.5]$ be fixed; this will be our target value for $\alpha_L$.   Now we introduce a program which aims to check if $\alpha_L = \alpha$ is satisfies \eqref{inequ:individual-alpha-L}.  In the program, we have a variable $y'_g$ for every size-configuration $g$, and three variables $V_0, V_1$ and $V_2$. The program is defined by the objective \eqref{equ:program}, and constraints (\ref{equ:program-constraint:1-configuration}-\ref{equ:program-constraint:define-vk}).\footnote{Notice the program has uncountable many variables, but \eqref{equ:program-constraint:1-configuration} and \eqref{equ:program-constraint:non-negative} will require there are countably many variables with non-zero values.}
	\begin{equation}
		\max \qquad \sum_g y'_{g} \Big(\big(1-\frac\alpha 2\big) \sum_{a \in g} a^2 - \frac\alpha2 \big(\sum g\big)^2 \Big) + \frac12 L^2  - \frac12 \big((\min(V_0, 1))^2 + (\min (V_1, 2))^2 + V_2^2\big)  \label{equ:program}
	\end{equation}\vspace*{-20pt}

	\noindent\begin{minipage}[t]{0.48\textwidth}
		\begin{align}
			\sum_{g} y'_{g} &= 1 \label{equ:program-constraint:1-configuration} \\
			y'_{g} &\geq 0 &\quad &\forall g \label{equ:program-constraint:non-negative}
		\end{align}
	\end{minipage}\hfill
	\begin{minipage}[t]{0.48\textwidth}
		\begin{align}
			\sum_{g} y'_{g}\cdot (\sum g) &= L \label{equ:program-constraint:define-L}\\
			\sum_{g} y'_{g} v_k(g) - V_k &= 0 &\quad &\forall k \in \{0, 1, 2\} \label{equ:program-constraint:define-vk}
		\end{align}
	\end{minipage}\medskip
	
	In the program, $y'_g$ is the fraction of $g$ we choose.   \eqref{equ:program-constraint:1-configuration} requires that we choose in total 1 fraction of size-configuration, \eqref{equ:program-constraint:non-negative} is the non-negativity condition. \eqref{equ:program-constraint:define-L} gives the definition of $L$. $V_0, V_1$ and $V_2$ are the total ``volume'' of elements in $[1, 2), [2, 4)$ and $[4, 8)$ respectively; this gives us \eqref{equ:program-constraint:define-vk}.  Notice that \eqref{inequ:individual-alpha-L}  is equivalent to $\displaystyle \Big(1 - \frac{\alpha_L}{2}\Big)Q - \frac{\alpha_L}{2} \cdot F + \frac12 \cdot L^2 -\frac12 \cdot S \leq 0$.   This gives us the objective \eqref{equ:program}. 
	
	\begin{lemma}
		If for some $\alpha \in [1, 1.5]$, program~\eqref{equ:program} has value at most $0$, then \eqref{inequ:individual-alpha-L} is satisfied with $\alpha_L = \alpha$. 
	\end{lemma}
	\begin{proof}
		Consider the $y_{if}$ values for the configurations $f$. We can covert the configurations to size configurations naturally: Start from $y'_g = 0$ for every size-configuration $g$.  For every configuration $f$ with $y_{if} > 0$, we let $g = \{q_j: j \in f\}$ be its correspondent size-configuration ($g$ is a multi-set). Then we increase $y'_g$ by $y_{if}$.  Let $V_k = \sum_{g} y'_g v_k(g)$ for every $k \in \{0, 1, 2\}$. Clearly, $(y', V_0, V_1, V_2)$ satisfy all the constraints in the program.   So the value of \eqref{equ:program} on this $(y', V_0, V_1, V_2)$ is at most $0$. Then, 
		\begin{align*}
			&\quad \Big(1 - \frac{\alpha}{2}\Big)Q - \frac{\alpha}{2} \cdot F + \frac12 \cdot L^2 -\frac12 \cdot S\\
			&= \sum_{f\subseteq [j^*]} y_{if} \left(\big(1-\frac\alpha2\big)\sum_{j\in f} q_j^2 - \frac\alpha2\cdot q^2(f) + \frac 12 L^2\right) - \frac12\sum_{k}\Big(\min\Big\{\sum_f y_{if} \cdot v_k(f), 2^k\Big\}\Big)^2\\
			&\leq \sum_{g} y'_{g} \left(\big(1-\frac\alpha2\big)\sum_{a\in g} a^2 - \frac\alpha2\cdot \big(\sum g\big)^2 + \frac 12 L^2\right) - \frac12\sum_{k=0}^2\Big(\min\Big\{\sum_g y'_{g} \cdot v_k(g), 2^k\Big\}\Big)^2\\
			&\leq \sum_{g} y'_{g} \left(\big(1-\frac\alpha2\big)\sum_{a\in g} a^2 - \frac\alpha2\cdot \big(\sum g\big)^2 + \frac 12 L^2\right) - \frac12 \big((\min(V_0, 1))^2 + (\min (V_1, 2))^2 + V_2^2\big)  \leq 0.
		\end{align*}
		Notice that $V_2 \leq \sum_{f}y'_{g} v_2(f) \leq \sum_{f}y'_{g} q(f) = L < 4$  and thus $\min(V_2, 4) = V_2$. 
		The inequality is equivalent to $Q + \frac12 L^2 - \frac12 S \leq \alpha\cdot \frac12(Q + F)$; so setting $\alpha_L= \alpha$ will satisfy \eqref{inequ:individual-alpha-L}. 
	\end{proof}
		
	Next, we show that we only need to consider the size-configurations $g$ satisfying certain property:
	\begin{lemma}
		\label{lemma:restrict-f}
		In program~\eqref{equ:program},  we can w.l.o.g assume $y'_{g} = 0$ for any size-configuration $g$ with $\big(\sum g\big) - (\min g) \geq L$. 
	\end{lemma}
	\begin{proof}
		If $y'_{g} > 0$ for such a size-configuration $g$. We let $a = \min g$, we can then remove $a$ from $g$, and add the $y'_g$ fraction of $a$ to $y'_g$ fractional size-configurations $g'$ with $(\sum g') < L$. This can be done as $\sum_{g} y'_{g} \sum_{a \in g} a = L$.  This operation does not change $V_0, V_1$ and $V_2$, and $\sum_{g}y'_{g} \sum_{a \in g}a^2$, and it will decrease $\sum_f y'_{g} (\sum g)^2$.  Therefore, it will increase the value of the objective. 
	\end{proof}
	
	Program~\eqref{equ:program} is not a convex program, due to the terms $(\min(V_0, 1))^2$ and $(\min(V_1, 2))^2$. Therefore, we break the program into 3 sub-programs, each of which is a convex one. They all have the same set of constraints (\ref{equ:program-constraint:1-configuration}-\ref{equ:program-constraint:define-vk}).  The objectives of the 3 sub-programs are respectively
	\begin{align}
		&\max \quad \sum_g y'_{g} \Big(\big(1-\frac\alpha 2\big) \sum_{a \in g} a^2 - \frac\alpha2 \big(\sum g\big)^2 \Big) + \frac12 L^2  - \frac12 \cdot 2^2, \label{equ:sub-program-1}\\
		&\max \quad \sum_g y'_{g} \Big(\big(1-\frac\alpha 2\big) \sum_{a \in g} a^2 - \frac\alpha2 \big(\sum g\big)^2 \Big) + \frac12 L^2  - \frac12 \big(1^2 + V_1^2 + V_2^2\big), \label{equ:sub-program-2}\\
		&\max \quad \sum_g y'_{g} \Big(\big(1-\frac\alpha 2\big) \sum_{a \in g} a^2 - \frac\alpha2 \big(\sum g\big)^2 \Big) + \frac12 L^2  - \frac12 \big(V_0^2 + V_1^2 + V_2^2\big). \label{equ:sub-program-3}
	\end{align}	
	Notice we do not need to add constraints $V_1\leq 2$ or $V_1\geq 2$ to the sub-programs. 
	\begin{lemma}
		If for some $\alpha \in [1, 1.5]$, the values of the three sub-programs \eqref{equ:sub-program-1}, \eqref{equ:sub-program-2} and \eqref{equ:sub-program-3} are non-positive, so is the value of program \eqref{equ:program}. 
	\end{lemma}
	\begin{proof}
		We consider the contra-positive of the statement: We assume the value of program \eqref{equ:program} is positive, and prove that one of the three sub-programs have positive value.  Consider a solution $(y', V_0, V_1, V_2)$ to the program \eqref{equ:program} that gives the positive value. If $V_1 \geq 2$, then the objective \eqref{equ:sub-program-1} is positive.  If $V_1 < 2$ and $V_0 \geq 1$, then the objective becomes \eqref{equ:sub-program-2}, and thus it is positive. Finally, if $V_1  < 2$ and $V_0 < 1$, the objective \eqref{equ:sub-program-3} is positive. 
	\end{proof}
	
	Thus our goal is now to set $\alpha_L$ values, so that the three sub-programs have non-positive values. 
	\subsection{Analysis of Sub-Program \eqref{equ:sub-program-1}.} First we consider sub-program \eqref{equ:sub-program-1}. For any feasible solution $(y', V_0, V_1, V_2)$, the objective of the program is
	\begin{align*}
		&\quad \sum_f y'_{g} \Big(\big(1-\frac\alpha 2\big) \sum_{a \in g} a^2 - \frac\alpha2 \big(\sum g\big)^2 \Big) + \frac12 L^2  - \frac12 \cdot 2^2\\
		&\leq - \sum_{g} y'_{g} \cdot (\alpha - 1) \big(\sum g\big)^2 + \frac12 L^2  - 2 \leq -(\alpha - 1) L^2 + \frac12 L^2 - 2 = \Big(\frac32 - \alpha\Big) L^2 - 2.
	\end{align*}
	We used that $\sum_{a \in g} a^2 \leq \big(\sum g\big)^2$ and $L = \sum_g y'_g(\sum g)$. So, for $\alpha = \frac32 - \frac{2}{L^2}$, the value of the sub-program is non-positive. 
	
	\subsection{Analysis of Sub-Program \eqref{equ:sub-program-2}.} Now, we consider sub-program~\eqref{equ:sub-program-2}.  As the objective does not contain $V_0$, constraint \eqref{equ:program-constraint:define-vk} for $k=0$ becomes redundant.  Using Lagrangian multipliers, the value of the sub-program is at most
	\begin{align}
		\inf_{\mu, \lambda_1, \lambda_2 \in \R_{\geq 0}} \quad \sup_{y', V_1, V_2} \qquad & \sum_g y'_{g} \Big(\big(1-\frac\alpha 2\big) \sum_{a \in g} a^2 - \frac\alpha2 \big(\sum g\big)^2 \Big) + \frac12 L^2  - \frac12 \big(1 + V_1^2 + V_2^2\big) \nonumber \\
		& + \mu \big(\sum_{g} y'_{g} \big(\sum g\big) - L\big)  - \lambda_1 \big(\sum_{g} y'_{g} v_1(g) - V_1\big) - \lambda_2\big(\sum_{g} y'_{g} v_2(g) - V_2\big). \label{equ:Langragian-1}
	\end{align}
	Under the sup operator,  $y'$ is over all vectors such that $\sum_g y'_{g} = 1$ and $y'_{g} \geq 0, \forall g$, and $V_1, V_2 \in \R$.   Notice that we do not require \eqref{equ:program-constraint:define-L} and \eqref{equ:program-constraint:define-vk} to hold. Also, we could allow $\mu, \lambda_1$ and $\lambda_2$ to be in $\R$. We chose the signs to be $+$, $-$, and $-$ before them, and restrict them in $\R_{\geq 0}$; this only increases the quantity. 
	
	For any for fixed $\mu, \lambda_1, \lambda_2 \in \R_{\geq 0}$ and $y'$, the objective \eqref{equ:Langragian-1} is maximized when $V_1 = \lambda_1$ and $V_2 = \lambda_2$. Also,  for fixed $\mu, \lambda_1, \lambda_2, V_1$ and $V_2$, the objective is linear in $y'$.  So it is maximized when $y'$ is a vertex point of the simplex. That is $y'_{g} = 1$ for some size-configuration $g$, and $y'_{g'} = 0$ for every $g' \neq g$.   Therefore, \eqref{equ:Langragian-1} is equal to 
	\begin{align}
		\inf_{\mu, \lambda_1, \lambda_2 \in \R_{\geq 0}} \quad \sup_g \quad  &\big(1 - \frac\alpha 2\big) \sum_{a \in g}a^2 - \frac\alpha 2 \big(\sum g\big)^2 + \mu \big(\sum g\big) - \lambda_1 v_1(g) - \lambda_2 v_2(g) \nonumber\\
		& + \frac12 L^2  - \mu L- \frac12 + \frac12(\lambda_1^2 + \lambda_2^2). \label{equ:Langragian-1a}
	\end{align}
	
	To handle the issues raised by the open intervals, we redefine  $v_1(g)$ and $v_2(g)$ in a slightly different way (the definition of $v_0(g)$ is irrelevant in this case). First, we assume in the definition of size-configurations, each $a \in g$ is associated with one of the 4 types: type-s, type-1, type-2 and type-3.  The following conditions must be satisfied:
	\begin{itemize}
		\item If $a$ is of type-s, then $a \in (0, 2]$;
		\item If $a$ is of type-1, then $a \in [2, 4]$;
		\item If $a$ is of type-2, then $a \in [4, 8]$;
		\item If $a$ is of type-b, then $a \geq 8$. 
	\end{itemize}
	Then we redefine $v_1(g)$ and $v_2(g)$ as the total value of type-1 and type-2 elements in $g$ respectively. Now, the element $2 \in g$ can contribute to either $v_1(g)$ or $v_2(g)$.  This slight modification will not change the value of \eqref{equ:Langragian-1a}.
	
%	So, \eqref{equ:Langragian-1a} is at most 
%	\begin{align}
%		\inf_{\mu, \lambda_1, \lambda_2 \in \R} \quad \sup_g \quad  &\big(1 - \frac\alpha 2\big) \sum_{a \in g}a^2 - \frac\alpha 2 \big(\sum_{a \in g}a\big)^2 + \mu \sum_{a \in g}a - \lambda_1 v_1(g) - \lambda_2 v_2(g) \nonumber\\
%		& + \frac12 L^2  - \mu L- \frac12 + \frac12(\lambda_1^2 + \lambda_2^2), \label{equ:Langragian-1a}
%	\end{align}
%	where $g$ is over all multi-sets of positive reals of size at most $n$.  We made the intervals closed. If $a = 2$, then $a$ could be of type-1 or type-2. 
	
	By Lemma~\ref{lemma:restrict-f}, we can restrict $g$ to satisfy the following property: the sum of the numbers in $g$, excluding the smallest one, is less than $L$.  We say a number $a > 0$ is flexible if $a \notin \{2, 4, 8\}$. We then prove 
	\begin{lemma}
		In \eqref{equ:Langragian-1a}, we can restrict ourselves to the multi-sets $g$ containing at most 1 flexible element. 
	\end{lemma}
	\begin{proof}
		Notice that if $g$ contains one element that is at least $4 > L$, then it is the only element in $g$.  Assume there are two flexible elements $b, b' \in (0, 4)$ in $g$. There are 4 numbers $l, r, l', r'$ such that the following conditions hold.
		\begin{itemize}
			\item $l < b < r$ and $l' < b' < r'$.
			\item The three numbers $l, b$ and $r$ are all in $[0, 2]$, or all in $[2, 4]$. This also holds for $l', b'$ and $r'$. 
			\item At least one of $l$ and $r'$ is in $\{0, 2, 4\}$. This also holds for $r$ and $l'$. 
			\item $l + r' = b + b' = r + l'$. 
		\end{itemize}
		Therefore, $(b, b')$ is a convex combination of $(l, r')$ and $(r, l')$. Say $(b, b') = \gamma(l, r') + (1-\gamma)(r, l')$ for some $\gamma \in [0, 1]$.  By changing $\{b, b'\}$ to $\{l, r'\}$ and $\{r, l'\}$ respectively, we obtain two multi-sets $g'$ and $g''$.  Moreover, we can make the type of $l$ and $r$ the same as that of $b$, and the type of $l'$ and $r'$ the same as that of $b'$.  $g'$ and $g''$ may contain $0$, but this is not an issue as removing $0$ from the set does not change the objective. 
		
		Then we have $\sum_{a \in g} a^2 < \gamma \sum_{a \in g'} a^2 +(1- \gamma) \sum_{a \in g''} a^2$, $(\sum g) = (\sum g') = (\sum g'')$, $v_1(g) = \gamma v_1(g') + (1-\gamma) v_1(g'')$ and $v_2(g) = \gamma v_2(g') + (1-\gamma) v_2(g'')$. As $1 - \frac\alpha 2 \geq 0$, one of $g'$ and $g''$ has a larger objective than $g$ in \eqref{equ:Langragian-1a}, for any $\mu, \lambda_1, \lambda_2 \in \R_{\geq 0}$.  Therefore, we can disregard $g$. 
	\end{proof}
	
	\begin{figure}
		\centering
		\includegraphics[width=0.6\textwidth]{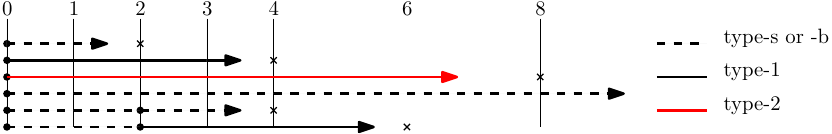}
		\caption{The 6 cases for Lemma~\ref{lemma:sub-program2-cases}. Each row of line segments represents a case in Lemma~\ref{lemma:sub-program2-cases}. Each line segment between two dots represent an element whose value is the length of the line segment. Each ray represents a flexible element (whose length is not determined) and the cross represents the rightmost endpoint it can reach. 
		 Dashed line segments represent type-s or type-b elements, black solid line segments represent type-1 elements, and red solid line segments represent type-2 elements.}
		\label{fig:configuration-1}
	\end{figure}
	
	\begin{lemma}
		\label{lemma:sub-program2-cases}
		To compute in \eqref{equ:Langragian-1a}, it suffices to consider the one of the following  6 cases for $g$ (See Figure~\ref{fig:configuration-1}):
		\begin{enumerate}[label=\arabic*)]
			\item $g = \{a_1\}$. $a_1$ is in $(0, 2]$ of type-s, in $[2, 4]$ of type-1, in $[4, 8]$ of type-2,  or in $[8, \infty)$ of type-b. There are 4 cases here. 
			\item $g = \{a_1, a_2\}$. $a_1$ is $2$ of type-s. $a_2$ is in $[0, 2]$ of type-$s$, or in $[2, 4]$ of type-$1$. There are 2 cases here.
		\end{enumerate}
	\end{lemma}
	\begin{proof}
		Recall the two properties we can impose on $g$. The sum of elements in $g$, excluding the smallest element, is less than $L  < 4$. $g$ contains at most one element outside $\{2, 4, 8\}$.  Moreover, if $a = 2 \in g$, we can assume $a$ is of type-s; if $a = 8 \in g$, we can assume $a$ is of type-b. This can only increase the objective.  Simple enumeration gives the 6 cases.  %TODO \lambda_1 \geq 0, \lambda_2 \geq 0 
	\end{proof}
	
	%TODO: or use $g$?  Think about this later. 
	
%TODO:	Restrict $f$: after removing the shortest job, the length becomes less than $L$.  Before.
	
	%TODO: change the order for configurations: from the smallest cardinality to the biggest. 
	
	\subsection{Analysis of Sub-Program \eqref{equ:sub-program-3}} Finally, we consider sub-program \eqref{equ:sub-program-3}. Similar to the argument for sub-program~\eqref{equ:sub-program-2}, an element $a \in g$ is  one of the following 5 types: type-s, type-0, type-1,type-2 and type-b. Moreover, 
	\begin{itemize}
		\item If $a$ is of type-s, then $a \in (0, 1]$;
		\item If $a$ is of type-0, then $a \in [1, 2]$;
		\item If $a$ is of type-1, then $a \in [2, 4]$;
		\item If $a$ is of type-2, then $a \in [4, 8]$;
		\item If $a$ is of type-b, then $a \geq 8$. 
	\end{itemize}
	Redefine $v_0(g), v_1(g)$ and $v_2(g)$ to be the sum of elements in $g$ of type-$0$, type-$1$ and type-$2$ respectively. Following a similar argument as before, the value of the program is at most 
	%Similarly, for a fixed $\mu, \lambda_0, \lambda_1, \lambda_2$, it is the best to set $V_0 = \lambda_0, V_1 = \lambda_1$ and $V_2 = \lambda_2$. So, we consider the following program: 
		\begin{align}
			\inf_{\mu, \lambda_0, \lambda_1, \lambda_2 \in \R_{\geq 0}} \quad \sup_g \quad  &\big(1 - \frac\alpha 2\big) \sum_{a \in g} a^2 - \frac\alpha 2 \big(\sum g\big)^2 + \mu \big(\sum g\big) - \lambda_0 v_0(g) - \lambda_1 v_1(g) - \lambda_2 v_2(g) \nonumber\\
			& + \frac12 L^2 - \mu L + \frac12(\lambda_0^2 + \lambda_1^2 + \lambda_2^2). \label{equ:Langragian-2a}
		\end{align}
	
	Similarly, we say $a$ is flexible if $a \notin \{1, 2, 4, 8\}$. We can again prove the following lemma:
	\begin{lemma}
		In \eqref{equ:Langragian-2a}, we can restrict ourselves to the multi-sets $g$ containing at most 1 flexible element. 
	\end{lemma}
	
	\begin{figure}
		\centering
		\includegraphics[width=0.6\textwidth]{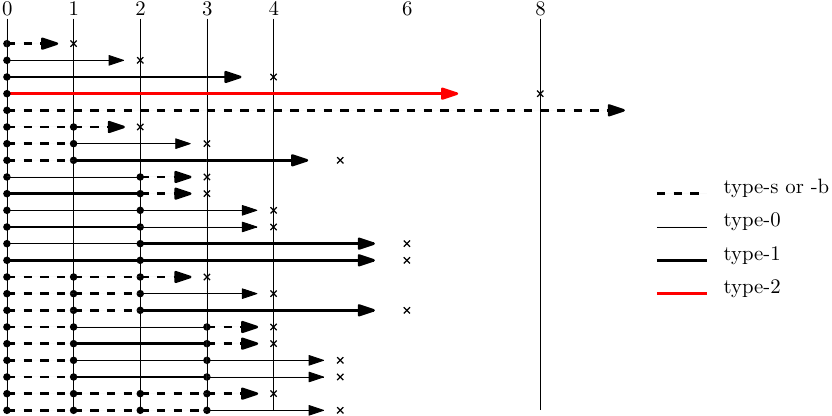}
		\caption{The 23 cases in Lemma~\ref{lemma:sub-program3-cases}.  The legends are the same as those in Figure~\ref{fig:configuration-1}, except now we have type-0 elements, denoted by thin solid line segments. }
	\end{figure}
	
	\begin{lemma}
		\label{lemma:sub-program3-cases}
		To compute in \eqref{equ:Langragian-2a}, it suffices to consider the one of the following  23 cases for $g$:
		\begin{enumerate}[label=\arabic*)]
			\item $g = \{a_1\}$. $a_1$ is in $(0, 1]$ of type-s,  in $[1, 2]$ of type-0, in $[2,4]$ of type-1, in $[4, 8]$ of type-2,  or in $[8, \infty)$ of type-b. There are 5 cases here. 
			\item $g = \{a_1, a_2\}$. $a_1$ is 1 of type-s, 2 of type-0, or 2 of type-1. $a_2$ is in $(0, 1]$ of type-s, in $[1, 2]$ of type-0, or in $[2, 4]$ of type-1. There are 9 cases here. 
			\item $g = \{a_1, a_2, a_3\}$. $a_1 = 1$ is of type-s. $a_2$ is $1$ of type-s, $2$ of type-0, or $2$ of type-1. $a_3$ is in $(0, 1]$ of type-s, or $[1, 2]$ of type-0.  There are 6 cases here.
			\item $g= \{a_1, a_2, a_3\}$. $a_1 = a_2 = 1$ of type-s, $a_3 \in [2, 4]$ of type-1.
			\item $g =\{a_1, a_2, a_3, a_4\}$. $a_1 = a_2 = a_3 = 1$ is of type-s, $a_2$ is in $(0, 1]$ of type-s, or in $[1, 2]$ of type-0. There are 2 cases here. 
		\end{enumerate}
	\end{lemma}
	\begin{proof}
		Again, we can require $g$ to satisfy the following two properties. The sum of elements in $g$ excluding the smallest one is less than $L < 4$.  There are at most 1 element in $g$ outside $\{1, 2, 4,8\}$. If $a = 1 \in g$, then $a$ is of type-s. Exhaustive enumeration gives the 23 cases.  
	\end{proof}
	
	\subsection{Parameters for sub-programs \eqref{equ:sub-program-2} and \eqref{equ:sub-program-3}} We break our interval $[2, 4)$ for $L$ into 10 sub-intervals, the $i$-th interval being $[2^{1+(o-1)/10}, 2^{1+o/10})$.  For each interval of $L$, we give the parameters $\mu, \lambda_1, \lambda_2 \in \R_{\geq 0}$ for \eqref{equ:Langragian-1a} and $\mu, \lambda_0, \lambda_1, \lambda_2 \in \R_{\geq 0}$ for \eqref{equ:Langragian-2a}, and the overall $\alpha_L$ in Table~\ref{table:parameters}.  Notice that $L$ has an equal probability of falling into the 10 intervals. $\E_\beta [\alpha_L]$ is the average of the $\alpha_L$ values over the 10 intervals, which is less than $1.358 < 1.36$. 
	
	\begin{table}[h]
		\centering
		\begin{tabular}{||c||c|c|c|c||c|c|c|c||c||}
			\hline
			\multirow{2}{*}{$L$} & \multirow{2}{*}{$\frac32-\frac2{L^2}$} &  \multicolumn{3}{c||}{Parameters for \eqref{equ:Langragian-1a}} & \multicolumn{4}{c||}{Parameters for \eqref{equ:Langragian-2a}} & \multirow{2}{*}{$\alpha_L$} \\\cline{3-9}
			 & & $\mu$ & $\lambda_1$ & $\lambda_2$ & $\mu$ & $\lambda_0$ & $\lambda_1$ & $\lambda_2$ & \\\hline
			$[2^1, 2^{1.1})$ & 1.0648 & 2.060000&0.08&0.00& 2.072000&0.3145&0.3364&0.0828&1.376228\\\hline
			$[2^{1.1}, 2^{1.2})$ & 1.1211  &2.293595&0.16&0.04&2.220500&0.3154&0.3642&0.1604&1.370445\\\hline
			$[2^{1.2}, 2^{1.3})$ & 1.1702 & 2.458214&0.22&0.16&2.508757&0.3178&0.4442&0.3200&1.364426\\\hline
			$[2^{1.3}, 2^{1.4})$ & 1.2129 &2.634649&0.32&0.28&2.720583&0.3220&0.5288&0.4532& 1.356049\\\hline
			$[2^{1.4}, 2^{1.5})$ & 1.2500 &2.823747&0.42&0.40&2.874152&0.3255&0.5942&0.5472&1.349022 \\\hline
			$[2^{1.5}, 2^{1.6})$ & 1.2824 &2.998133&0.48&0.48&2.961363&0.3279&0.6310&0.5984& 1.344238\\\hline
			$[2^{1.6}, 2^{1.7})$ & 1.3106 &3.213319&0.56&0.56&3.150265&0.3292&0.6886&0.6752&1.341530 \\\hline
			$[2^{1.7}, 2^{1.8})$ & 1.3351 &3.346480&0.60&0.60&3.343231&0.3296&0.7362&0.7336&1.340912\\\hline
			$[2^{1.8}, 2^{1.9})$ & 1.356413 &3.412558&0.60&0.60&3.404549&0.3273&0.7398&0.7380&1.356413\\\hline
			$[2^{1.9}, 2^2)$ & 1.3750 & 3.470883&0.60&0.60&3.507084&0.3009&0.7328&0.7324&1.375000\\\hline
		\end{tabular}
		\caption{The parameters for the three sub-programs \eqref{equ:sub-program-1}, \eqref{equ:sub-program-2} and \eqref{equ:sub-program-3}, for each interval of $L$. The first column contains the  intervals for $L$. The second column gives the $\frac32 - \frac{2}{L^2}$ for the right end point $L$ of each interval, which is the minimum $\alpha_L$ needed for sub-program \eqref{equ:sub-program-1}.  The next 3 columns give the parameters $\mu, \lambda_1, \lambda_2$ for \eqref{equ:Langragian-1a}. The next 4 columns give the parameters $\mu, \lambda_1, \lambda_2, \lambda_3$ for \eqref{equ:Langragian-2a}. The final column gives the $\alpha_L$ value for each interval of $L$. The average of the 10 values of $\alpha_L$ is less than $1.36$.} \label{table:parameters}
	\end{table}

	We then describe how to check if the parameters are valid using a computer program. Focus on each interval $[2^{1+(o-1)/10}, 2^{1+o/10})$ for $L$. For sub-program \eqref{equ:sub-program-1}, we can simply check if $\alpha_L\geq \frac32 - \frac 2{L^2}$ for $L=2^{1+o/10}$; the values of $\frac32 - \frac 2{L^2}$ for $L=2^{1+o/10}$ are given in the table.  For sub-program \eqref{equ:sub-program-2}, we consider \eqref{equ:Langragian-1a}, for the given $\mu, \lambda_1, \lambda_2, \alpha = \alpha_L$ and each case of $g$ in Lemma~\ref{lemma:sub-program2-cases}.  For a fixed $g$ in the case, the objective is a quadratic function of $L$ with the quadratic term being $\frac{L^2}{2}$. Therefore, the worst $L$ is one of the two end points $2^{1+(o-1)/10}$ and $2^{1+o/10}$ of the interval.  So, we only need to check the two values for $L$. Fixing $L$ and letting $a$ be only flexible element in $g$, the quantity in \eqref{equ:Langragian-1a} is a quadratic function $a$.  Therefore, we can easily find the maximum of the objective over all valid $a$'s for the case.  Similarly, we can check the validity of the parameters for \eqref{equ:sub-program-3} using \eqref{equ:Langragian-2a}.    Using our checker program with the given parameters, we can certify that the  $\alpha_L$ values are feasible for \eqref{inequ:individual-alpha-L}.  The source code for the checker can be found at  
	\begin{itemize}
		\item \url{https://github.com/shili1986/WeightedCompletionTime.git}. 
	\end{itemize} \medskip

	We also describe briefly how we search for the parameters $\mu, \lambda_1, \lambda_2$ for \eqref{equ:Langragian-1a} and $\mu, \lambda_0, \lambda_1, \lambda_2$ for \eqref{equ:Langragian-2a}, though this is not needed to for the proof of our approximation ratio.  We only describe our parameter searching algorithm for \eqref{equ:Langragian-1a}; the procedure for \eqref{equ:Langragian-2a} uses similar ideas.  We enumerate $\mu, \lambda_1$ and $\lambda_2$, each parameter having a lower bound, an upper bound and a step size. That means, we search for the optimum vector $(\mu, \lambda_1, \lambda_2)$ within a box using some step-sizes. For a fixed $\mu, \lambda_1$ and $\lambda_2$, the best $\alpha_L$ can be determined using binary search up to any precision.    Then, once we find the optimum $(\mu, \lambda_1, \lambda_2)$, we run a fine-grained searching procedure in a smaller box around the $(\mu, \lambda_1, \lambda_2)$ we found, with smaller step sizes for the three parameters. Once we find the optimum $(\mu', \lambda'_1, \lambda'_2)$, we run a third round of even more fine-grained search procedure around it, which will give our final parameters for $\mu, \lambda_1$ and $\lambda_2$.  We also tried different values of $\rho$, but it seems that the gain is small. Thus we decided to use the simple choice of $\rho = 2$.

\section{Discussions}  \label{sec:discuss}
In this paper, we developed a $1.36$-approximation algorithm for the unrelated machine weighted completion time scheduling problem, using an iterative rounding procedure.  Unlike previous algorithms, which impose the non-positive correlation requirement between any two edges $ij$ and $ij'$, we only require non-positive correlation to hold between an unmarked edge $e \in \delta^{k, \umk}_i$ and the set $\delta^{k, \mk}_i$ of marked edges as a whole. With this relaxation, our algorithm chooses at most one edge from each group $\delta^{k, \mk}_i$, leading to the upper bound in  \eqref{inequ:int-cost-on-i}.  

An immediate open problem is to provide a tight upper bound on the ratio between the bounds in \eqref{inequ:int-cost-on-i} and \eqref{equ:LP-cost-for-i}, for a fixed $i$ and $j^*$.  Our analysis falls short of achieving this goal for two reasons. Firstly, we only consider three job classes $J_k$, as we have to guess whether $V_k \geq 2^k$ or not for every $k$ in program~\eqref{equ:program}. Increasing the number of job classes would require us to consider more sub-programs.  Secondly, we did not exploit the property that the $z$-vector is the same for every random value $\beta \in [1, \rho = 2)$. In our analysis, we define the parameter $L \in [2, 4)$, and compute the worst ratio for every $L$. However, the worst case scenarios for different $L$ values may be inconsistent. Exploring the connections between different $L$ values would lead to a tighter analysis.  Additionally, it would be interesting to give a clean analytical proof of the ratio, without significantly degrading the approximation performance. 

%Finally, it is also interesting to obtain a simple analytical analysis of the approximation ratio without the assistance of computer programs.  Offload the burden of analyzing the ratio to a computer program. 

%There is still a gap in our analysis.  What is the worst gap between xx and xx? In our analysis, we did not use the connection between different values of \beta. L.  Moreover, as the program is not convex, we need to enumerate all possible cases. We can only consider the three most important classes. 
%Is it possible to obtain an analytical proof?

\bibliographystyle{plain}
\bibliography{reflist}

\end{document}